%% file: main.tex
\documentclass[]{JFM-FLM_Au}

\input{dave_notation}
\crefname{equation}{}{}

\newcolumntype{M}{>{\begin{varwidth}{4cm}}l<{\end{varwidth}}} %M is for Maximal column

\lefttitle{D. Darrow, E. Carlson, \& D. Goluskin}

\title{Quartic Lyapunov functions for global fluid stability}

\author{David Darrow\aff{1}, Elizabeth Carlson\aff{2}, David Goluskin\aff{3}}

\affiliation{\aff{1}Department of Mathematics, Massachusetts Institute of Technology, Cambridge, MA, 02139, USA
	\aff{2} Department of Mathematics, Oregon State University, Corvallis, OR, 97331, USA
	\aff{3} Department of Mathematics and Statistics, University of Victoria, Victoria, BC, V8P 5C2,
	Canada} 

\corresau{David Darrow, \email{ddarrow@mit.edu}}

\begin{document}
	\maketitle
	
	\begin{abstract}
 A fluid system is `globally stable' if all initial conditions eventually converge to the same state. Since \citet{reynolds_iv_1895} and \citet{orr_stability_1907}, the standard way to show global stability has been the energy method, which employs the fluctuation energy of the flow as a Lyapunov function. However, the energy method fails whenever transient energy growth is possible, so it often yields overly strict stability criteria. The first broadly applicable approach that can overcome this limitation has recently been introduced \citep{goulart_global_2012,fuentes_global_2022}, using polynomial optimization to numerically construct non-quadratic Lyapunov functions of certain forms. Unlike the energy method, however, this approach is highly technical, computationally expensive, and has been hard to interpret physically. Moreover, it treats only one set of parameters at a time; in particular, if it verifies global stability for a given system at a certain value of the Reynolds number, it does not imply the same for all smaller values. The present work makes progress by connecting this numerical program with new analytical and physical insights. We first show how to exploit symmetries of shear flows via convenient complex variable representations, greatly reducing the problem size. We then refine key inequalities in the previous works, replace several expensive computational steps with simpler analytical alternatives, and show how to prove global stability over a continuous range of Reynolds numbers. Our analysis reveals the simplest possible class of non-quadratic Lyapunov functions for two-dimensional parallel shear flows: a particular three-parameter family of quartic polynomials. Using these Lyapunov functions, we verify global stability of two-dimensional plane Couette flow and plane Poiseuille flow up to modestly higher Reynolds numbers than possible with the energy method. We explain---mathematically and physically---how Lyapunov functions in this class are able to decay monotonically despite transient energy growth. Our work takes a step towards a fully analytical theory of global fluid stability beyond the energy method, and it offers new structural insights that should significantly improve future numerical investigations of global stability.
	\end{abstract}

	\begin{keywords}
		%Authors should not enter keywords on the manuscript, as these must be chosen by the author during the online submission process and will then be added during the typesetting process (see \href{https://www.cambridge.org/core/journals/journal-of-fluid-mechanics/information/list-of-keywords}{Keyword PDF} for the full list).  Other classifications will be added at the same time.
	\end{keywords}
	
	%{\bf MSC Codes }  76A20; 76D45; 76S05

	\section{Introduction}\label{sec:intro}
	
	A fluid system is globally (asymptotically) stable if, no matter how or how strongly the fluid is initially perturbed, it always returns to the same base state. There are limited tools available to prove that a given system is globally stable. For one, standard techniques of local stability analysis are insufficient to prove global stability; linear instability of any steady flow precludes global stability, but linear stability says nothing either way. The laminar states of three-dimensional (3-D) plane Couette flow and pipe flow are linearly stable at all Reynolds numbers $\nRe$, for instance, but both systems lose global stability at $\nRe$ values smaller than $10^3$ \citep{waleffe_homotopy_2003,Pringle_2007}.
	
	Until recently, the only approach available to prove global stability was the `energy method' \citep{joseph_stability_1976}, first introduced by \citet{reynolds_iv_1895} and \citet{orr_stability_1907}. The energy method investigates the fluctuation energy of the flow, i.e., the energy of deviations $\vec{u}(\vec{x},t)$ from the base state, expressed as the following volume integral:
	\begin{equation}\label{eq:energy_primitive}
		E[\vec{u}] = \tfrac{1}{2}\|\vec{u}\|^2 = \tfrac{1}{2}\int |\vec{u}|^2.
	\end{equation}
    If one can show that $E$ decreases monotonically in time from any nonzero initial state, then global stability follows. For nearly any incompressible fluid system, the energy method verifies global stability up to a system-dependent energy Reynolds number $\nRe_E>0$ \citep{serrin_stability_1959}. In other words, the energy method yields a lower bound
    \begin{equation}\label{eq:ReRg}
\nRe_E<\nRe_G
    \end{equation}
on the global stability Reynolds number $\nRe_G$, i.e., the smallest Reynolds number at which the flow is not globally stable.
    
Unfortunately, the energy method's criterion for success---that no initial condition can lead to transient growth of $E$---is often overly strict, so the resulting bound \cref{eq:ReRg} is very conservative for many systems. For instance, in 3-D plane Couette flow with unit layer thickness, the energy method yields $\nRe_E \approx 82.6$ \citep{joseph_nonlinear_1966}, but non-laminar stationary states have been found only when $\nRe\gtrsim 511$, suggesting that $\nRe_G\approx511$ \citep{waleffe_homotopy_2003}. A similar gap exists for other parallel shear flows, including pipe flow and Poiseuille flow \citep{joseph_stability_1969,waleffe_homotopy_2003,Pringle_2007}. In 2-D plane Poiseuille flow with unit layer thickness, the energy method proves global stability below $\nRe_E\approx175.1$, while non-laminar stationary states have been found only when $\nRe\gtrsim 5878$ \citep{Casas_2012}. Strikingly, in 2-D plane Couette flow, the energy method shows global stability up to $\nRe_E \approx 177.2$ \citep{orr_stability_1907}, but extensive numerical investigations have failed to find sustained non-laminar solutions at any $\nRe$, suggesting that $\nRe_G=+\infty$ \citep{orszag_transition_1980,rincon_existence_2007,ehrenstein_two-dimensional_2008}.
	
One can hope to surpass the energy method by replacing $E[\vec{u}]$ with a different Lyapunov function---i.e., a different functional $V[\vec{u}]$ that is positive and decays monotonically from any nonzero initial state. If such a Lyapunov function exists at a given $\nRe$, the system is globally stable at that Reynolds number \citep{mironchenko_non-coercive_2019}. The energy method amounts to selecting $V[\vec{u}]=E[\vec{u}]$, which is a Lyapunov function only when $\nRe<\nRe_E$.
	
Even for simple fluid systems, however, it is very challenging to construct any Lyapunov function other than $E[\vec{u}]$ itself. The key element that makes the energy method tractable is that both $E[\vec{u}]$ and its time derivative $\dot{E}[\vec{u}]$ are quadratic functionals of $\vec{u}$; no cubic terms arise in $\dot{E}[\vec{u}]$ because $E$ is conserved by nonlinear advection. The maximum energy growth rate $\lambda_1=\max\dot{E}[\vec{u}]/E[\vec{u}]$ can thus be computed by solving a linear eigenproblem, and the energy method succeeds if and only if $\lambda_1<0$. 
    
Other quadratic functionals can serve as Lyapunov functions only in special cases. If one restricts to perturbations that are uniform in one direction, for instance, then a reweighted energy can be used in the energy method \citep{joseph_stability_1976,galdi_new_1990,straughan_energy_2003,kaiser_generalized_2005}. In 2-D flows with no boundaries or with fixed-stress boundary conditions, the enstrophy can also serve as a Lyapunov function \citep{iudovich_example_1965}. 
% It is currently unknown whether enstrophy can act as a Lyapunov function under no-slip boundary conditions \citep{Fraternale2018}, but, even if so, viscous vorticity generation at the boundary renders it comparatively difficult to analyze. 
For typical incompressible fluid systems, however, such quantities are not valid Lyapunov functions at any $\nRe$, and the fluctuation energy $E[\vec{u}]$ is the only possible quadratic Lyapunov function.
	
	%To move beyond the energy method and its variants, then, one has to make use of a non-quadratic functional $V[\vec{u}]$. Finding global extrema of a non-quadratic functional of $\vec{u}$ is not generally tractable, even numerically, so it is difficult to verify that a given non-quadratic functional satisfies $\dot{V}[\vec{u}]<0$. As such, the energy method has generally remained the only feasible route to proving global stability.
	
\cite{goulart_global_2012} proposed a computational strategy to construct non-quadratic Lyapunov functions for fluids, a problem that had appeared intractable. Briefly, one fixes a particular set of velocity modes $\vec{u}_1,...,\vec{u}_N$ and decomposes the velocity field as
	\begin{equation}
		\vec{u}(\vec{x},t) = \sum_{j=1}^N a_j(t)\vec{u}_j(\vec{x}) + \vec{v}(\vec{x},t).
	\end{equation}
Letting $q^2 = \tfrac{1}{2}\|\vec{v}\|^2$ denote the energy of the `tail' $\vec{v}$, one seeks a Lyapunov function of the form
	\begin{equation}
    \label{eq:goulart_ansatz}
		V[\vec{u}] = E[\vec{u}]^2 + P(a_1,...,a_N,q^2),
	\end{equation}
where $P$ is a cubic polynomial. Since $E=\frac{1}{2}|\vec{a}|^2+q^2$, the full functional is a quartic polynomial in the $N+1$ variables $(a_j,q)$, so there exist computational tools (discussed shortly) for enforcing $V\geq 0$. However, the time derivative $\dot{V}[\vec{u}]$ depends generally on the infinite-dimensional tail $\vec{v}$, making it intractable to directly check whether $V$ decreases in time for all possible $\vec{u}$. Instead, through detailed analysis of the particular system under consideration, one attempts to find a quartic polynomial $S=S(a_j,q^2)$ such that the bound
	\begin{equation}\label{eq:intro_polybound}
		\dot{V}[\vec{u}] = \dot{V}(a_1,...,a_N,\vec{v}) \leq S(a_1,...,a_N,q^2)
	\end{equation}
	holds for all $\vec{a}\in\mathbb{R}^N$ and all possible tails $\vec{v}$. At this point, global stability follows if one can show that $V(a_j,q^2)>0$ and $S(a_j,q^2)<0$ whenever $(a_j,q^2)\neq 0$. This is now a concrete, finite dimensional question: given a set of quartic polynomials in $N+1$ variables, can one verify that they are all positive away from zero? Although this problem is NP-hard as stated \citep{murty_np-complete_1987}, non-negativity can be enforced by stronger but more tractable sum-of-squares (SOS) conditions, which can be verified using semidefinite programming when $N$ is not too large \citep{parrilo_semidefinite_2003}.
	
The computational strategy of \citeauthor{goulart_global_2012} has yielded promising results so far. It was applied by \citet{huang_sum--squares_2015} to find a quartic Lyapunov function for streamwise-constant perturbations in a fully periodic version of rotating Couette flow, a system where a weighted energy already gives the exact value of $\nRe_G$. It was then refined by \cite{fuentes_global_2022} and applied to 2-D plane Couette flow at various choices of the streamwise period, yielding the first verification of global stability for any system beyond what can be shown by the energy method. The refined formulation was also applied to 2-D plane Poiseuille flow by \citet{iligaray_improved_2026} with a similar degree of success.  The computations of \citeauthor{fuentes_global_2022} and \citeauthor{iligaray_improved_2026} do not yet constitute computer-assisted proofs, but one could rigorously verify the numerics with the aid of interval arithmetic \citep{Kearfott1996,Parker_2024}.
% \footnote{In the linearized dynamics of 2D Couette about the laminar state, transient growth of enstrophy has been ruled out up to certain Re [cite], but that no implication of this result for the nonlinear dynamics has been shown, despite claims to the contrary by the authors of that study.}
	
There are several obstacles preventing the methods in \citet{fuentes_global_2022} and \citet{iligaray_improved_2026} from becoming a general purpose approach to global stability. For one, the required computations are expensive and must be repeated for each domain size and $\nRe$ value of interest. In particular, even if their algorithm verifies global stability at a given Reynolds number $\nRe^*$, it does not imply global stability for $\nRe<\nRe^*$, and thus does not directly yield a lower bound on $\nRe_G$. Computation time also grows rapidly with the number of modes $N$, and the required value of $N$ grows rapidly with $\nRe$.  As such, only mild improvements over the energy method have been possible thus far. Perhaps most importantly, the SOS algorithm is largely a black box from a physical viewpoint. Although the resulting functional $V[\vec{u}]$ is guaranteed to be a valid Lyapunov function, the polynomials $V[\vec{u}]$ and $S(a_j,q^2)$ together contain up to $O(N^4)$ monomial terms and are difficult to interpret. Furthermore, additional SOS conditions are used to verify the inequality \cref{eq:intro_polybound}, so even the relationship between $\dot{V}$ and $S$ is somewhat obfuscated.

    Here, we make progress by connecting this numerical program with new analytical and physical insights. We first introduce a simple schematic for quartic Lyapunov functions in general incompressible fluid systems, illustrating how potential energy reservoirs and local stability basins can conspire to control transient energy growth when $\nRe>\nRe_E$. We then discuss how system symmetries can be exploited to restrict the space of possible Lyapunov functions, as well as to reduce the computation cost of any resulting SOS calculations. We employ a convenient complex variable representation to carry this program out for 2-D shear flows, allowing us to fold streamwise translations and up-down symmetries into our notation and thereby greatly reduce the problem size. Connecting these symmetry considerations with our general schematic, we arrive at the simplest possible class of non-quadratic Lyapunov functions for 2-D planar shear flows:\begin{equation}\label{eq:lyapunov_primitive}
		V[\vec{u}] = E[\vec{u}]^2 + 2\gamma_0E[\vec{u}]\int \vec{u}\cdot\vec{u}_0 + \gamma_1E[\vec{u}] + 2\gamma_2\Re\left[\int \vec{u}\cdot\ovl{\vec{u}_1}\int\vec{u}\cdot\vec{u}_2\right],
	\end{equation}
where $\vec{u}_0$, $\vec{u}_1$, and $\vec{u}_2$ are particular complex-valued velocity modes, the parameters $\gamma_0,\gamma_1,\gamma_2\in\RR$ are chosen depending on the streamwise period, and $\Re[z]$ denotes the real part of a complex number $z$.

We then use Lyapunov functions of the class \cref{eq:lyapunov_primitive} to recover new global stability results for 2-D plane Couette and Poiseuille flows. In so doing, we develop interpretable, analytical alternatives to many of the key inequalities used by \citet{fuentes_global_2022}. In particular, we replace all of the expensive, matrix-valued SOS problems embedded in the inequality \cref{eq:intro_polybound} with appropriate versions of Young's inequality, and we prove that the replacement is no less sharp than the original black box approach in many cases (including all that arise in our analysis). Moreover, we show how to simultaneously verify global stability for all Reynolds numbers below a desired threshold, so our analysis yields true lower bounds on $\nRe_G$, rather than global stability only at particular Reynolds numbers. To our knowledge, these are the first lower bounds on $\nRe_G$ for any fluid system that are larger than $\nRe_E$.

Most steps of our analysis are done by hand, aside from numerical solution of a linear eigenproblem and a small SOS computation to verify positivity of a single polynomial. This procedure is more interpretable---and far less computationally demanding---than previous SOS computations. The new lower bounds we find on $\nRe_G$ for 2-D Couette flow are, depending on the streamwise period, up to $15\%$ larger than $\nRe_E$ and up to $3\%$ larger than the Reynolds numbers at which \citet{fuentes_global_2022} verified global stability with the same choice of velocity modes $\vec{u}_n$. Although our global stability results for particular flows only modestly surpass the energy method, our work paves the way for future progress on two fronts. First, for numerical approaches where one hopes to scale to many $\vec{u}_n$ modes, we offer ways to greatly reduce the computational cost relative to \citeauthor{fuentes_global_2022} and \citet{iligaray_improved_2026}. Second, our various observations and analytical derivations take a step towards a fully analytical theory of global stability beyond the energy method.
	
In \cref{sec:interpretation}, we review the mathematics of the energy method and introduce a general schematic for quartic Lyapunov functions in incompressible flows. We then explain how such functionals can control transient energy growth, illustrating the principle with a finite-dimensional toy model. In \cref{sec:preliminaries}, we formulate stability analysis of 2-D shear flows using vorticity--streamfunction variables. In \cref{sec:symmetries}, we prove a simple lemma restricting the space of candidate Lyapunov functions based on symmetry, leading to \cref{eq:lyapunov_primitive} as the simplest possible ansatz for 2-D shear flows. In \cref{sec:analysis}, we prove an analytical bound of the form \cref{eq:intro_polybound} for our simple Lyapunov function, and we use it to verify global stability for 2-D plane Couette flow and Poiseuille flow beyond $\nRe_E$ for a variety of streamwise periods. In \cref{sec:conclusion}, we discuss how our work may improve future numerical and analytical studies of global stability.
	
\section{Global stability beyond the energy method}
\label{sec:interpretation}

Consider a bounded 2-D or 3-D domain $\Omega$, which may have one or more periodic directions. For any stationary velocity field $\vec{U}(\vec{x})$ satisfying the incompressible Navier--Stokes equations on $\Omega$, any deviations $\vec{u}(\vec{x},t)$ from $\vec{U}$ evolve according to
	\begin{equation}
    \label{eq:NS_primitive}
		\vec{u}_t + (\vec{u}+\vec{U})\cdot\nabla\vec{u} + \vec{u}\cdot\nabla\vec{U} = -\nabla p + \nu\nabla^2\vec{u},\qquad \nabla\cdot\vec{u} = 0,
	\end{equation}
where $\nu=1/\nRe$ is a dimensionless viscosity. We assume $\vec{u}$ satisfies either no-slip or no-stress conditions on each part of the boundary $\partial\Omega$.

\subsection{The energy method}
\label{sec:energymethod}

We first review the energy method, whose results are used in our more involved approach. For convenience, we write $\langle a\rangle = \int a$ for the integral of any (scalar or vector) field $a$ over the domain $\Omega$, and we write 
\begin{equation}\label{eq:innerprod}\langle f,g\rangle = \langle fg\rangle,\qquad \langle \vec{u},\vec{v}\rangle =\langle\vec{u}\cdot\vec{v}\rangle,\qquad \langle \vec{T},\vec{R}\rangle = \langle \op{tr}(\vec{T}^t\vec{R})\rangle
\end{equation}
for the real $L^2$ inner product between fields of scalars, vectors, or matrices, respectively. Here and below, we write $\vec{T}^t$ for the real transpose of any matrix $\vec{T}$. 
    
If $\vec{u}(x,t)$ is a time-evolving solution to \cref{eq:NS_primitive}, then the instantaneous time derivative $\dot{E}[\vec{u}]$ of the energy can be computed as follows:
	\begin{equation}
		\dot{E}[\vec{u}] = \langle\vec{u},\vec{u}_t\rangle = \langle\vec{u},-(\vec{u}+\vec{U})\cdot\nabla\vec{u} - \vec{u}\cdot\nabla\vec{U} - \nabla p + \nu\nabla^2\vec{u}\rangle.
	\end{equation}
Integrating by parts and exploiting the incompressibility of $\vec{u}$ and $\vec{U}$, one can simplify this expression greatly:
	\begin{equation}\label{eq:edot_primitive}
		\dot{E}[\vec{u}] = -\langle\vec{u}\cdot\nabla\vec{U}\cdot\vec{u}\rangle - \nu\|\nabla\vec{u}\|^2.
	\end{equation}
Because the right-hand side of this equation is quadratic in the velocity field, one can bound it above by a scalar multiple of $E[\vec{u}]$:
	\begin{equation}\label{eq:edot_to_e}
		\dot{E}[\vec{u}] \leq 2\lambda_1E[\vec{u}],\qquad \lambda_1 = \max_{\substack{\vec{w}\in H^1_0\\\nabla\cdot\vec{w}=0}} \left[\frac{-\langle\vec{w}\cdot \nabla\vec{U}\cdot\vec{w}\rangle - \nu\|\nabla\vec{w}\|^2}{\|\vec{w}\|^2}\right],
	\end{equation}
where $H_0^1$ is a standard (Sobolev) space of vector fields satisfying $\vec{w}|_{\partial\Omega}=0$ and $\vec{w},\nabla\vec{w}\in L^2$. Crucially, although the inequality \cref{eq:edot_to_e} is stated in terms of time-evolving velocity fields $\vec{u}$, the maximization problem is defined in terms of time-independent vector fields $\vec{w}$, so its solution $\lambda_1$ is time-independent as well. The value of $\lambda_1$ can be calculated as the leading eigenvalue of the associated Euler--Lagrange equation,
	\begin{equation}\label{eq:eullag}
		\tfrac{1}{2}(\nabla\vec{U} + \nabla\vec{U}^t)\vec{w} - \nu\nabla^2\vec{w} +\nabla p = \lambda\vec{w}, \qquad \nabla\cdot\vec{w}=0,
	\end{equation}
where $p(\vec{x})$ is the Lagrange multiplier for the divergence-free constraint. On a bounded domain, there are countably infinite pairs $(\lambda_j,\vec{w}_j)$ that solve \cref{eq:eullag}, all real-valued. These solutions are known as the system's `energy eigenvalues' and `energy eigenmodes', respectively. With very mild regularity assumptions on $\vec{U}$ and $\Omega$, there is a unique maximum eigenvalue $\lambda_1$ for this system, and it is negative when $\nRe=1/\nu$ is sufficiently small. It increases as $\nRe$  increases, and there exists $\nRe_E>0$ such that $\lambda_1<0$ if and only if $\nRe<\nRe_E$. Applying Gr\"onwall's inequality \citep{ames1997inequalities} to the inequality in \cref{eq:edot_to_e} gives
	\begin{equation}
		E[\vec{u}(\cdot,t)] \leq E[\vec{u}(\cdot,0)] e^{2\lambda_1 t}
        \label{eq:E ineq}
	\end{equation}
	for any initial field $\vec{u}$. When $\nRe<\nRe_E$, the right-hand side of~\cref{eq:E ineq} decays to zero as $t\to\infty$, implying global stability.

	%Already, one can see the basic principles of the energy method at play in \cref{eq:edot_primitive}. Suppose $\vec{u}$ has a characteristic magnitude $u$ and a length scale $\ell$, and that $\vec{U}$ is of unit magnitude and of length scale $L$. Scaling the equation \cref{eq:edot_primitive}, we find that viscous dissipation should dominate the value of $\dot{E}[\vec{u}]$ whenever 
	%\[\ell\ll \sqrt{\nu L},\]
	%and thus that $\dot{E}[\vec{u}]<0$ in this regime. Of course, since $\Omega$ is bounded, we must have $\ell\lesssim L$. If one chooses $\nu$ sufficiently large (for instance, $\nu\gg \sqrt{L}$), then it follows that $\dot{E}[\vec{u}]<0$ for any velocity field $\vec{u}$. On the other hand, if energy is always monotonically decreasing (sufficiently quickly, at least), then sustained non-laminar solutions are impossible, and global stability follows.

	\subsection{Polynomial Lyapunov functions}
As discussed in \cref{sec:intro}, the energy Reynolds number $\nRe_E$ is often a very conservative lower bound on the Reynolds number $\nRe_G$ at which a given system loses global stability. By definition, $\nRe_E$ is the smallest Reynolds number at which some initial velocity field leads to energy growth. The energy method cannot distinguish, for instance, between transient energy growth and the onset of sustained turbulence. To improve upon the energy method, we will replace the energy with a more complicated Lyapunov function. 
    
    First, throughout this work, we say that $V[\vec{u}]$ is a `functional' if it maps spatial functions $\vec{u}(\vec{x})$ to $\RR$. We say that a functional $V[\vec{u}]$ is `polynomial' in $\vec{u}$ if it can be decomposed as
    \begin{equation}
        V[\vec{u}] = V_0 + V_1[\vec{u}] + V_2[\vec{u},\vec{u}] + \cdots + V_M[\vec{u},...,\vec{u}]
    \end{equation}
    for some $M\geq 0$, where each $V_m$ is linear in each of its arguments. We call such $V$ `quadratic', `cubic', or `quartic' if $M=2$, $3$, or $4$, respectively. Next, for any functional $V[\vec{u}]$, we define $\dot{V}[\vec{u}]$ as the so-called Lie derivative along the flow:
	\begin{equation}
    \label{eq:lie derivative}
    \dot{V}[\vec{u}] = \cL V[\vec{u}] = \tfrac{\d}{\d t}V[\tilde{\vec{u}}(\cdot,t)]\big|_{t=0},
    \end{equation}
	where $\tilde{\vec{u}}(\vec{x},t)$ solves \cref{eq:NS_primitive} with $\tilde{\vec{u}}(\vec{x},0)=\vec{u}(\vec{x})$. We will refer to the $\dot V$ functional as the `time derivative' of $V$ because it satisfies~\cref{eq:lie derivative}, but we emphasize that $V[\vec u]$ and $\dot V[\vec u]$ on their own are time-independent functionals of $\vec{u}$. For instance, $\dot{E}[\vec{u}]$ can be defined by the time-independent expression \cref{eq:edot_primitive}. %Time dependence comes only from considering $V$ along time-evolving dynamics.

    \begin{remark}\label{rem:polynomials}
        Because the Navier--Stokes equations have only linear and quadratic terms, taking the time derivative of any polynomial functional $V[\vec{u}]$ gives a functional $\dot{V}[\vec{u}]$ that is also polynomial with degree at most one higher than $V[\vec{u}]$. A crucial property of the energy is that $E[\vec{u}]$ and $\dot E[\vec{u}]$ are both quadratic. The nonlinear term of the Navier--Stokes equations does not produce a cubic term in $\dot E[\vec{u}]$ because $E[\vec{u}]$ is an Euler invariant \citep{Olver1993}, meaning it is conserved by nonlinear advection.
    \end{remark}
    
For a Lyapunov function to imply global stability, the exact conditions it must satisfy vary slightly between different types of governing dynamics \citep{mironchenko_non-coercive_2019}. Here, we define Lyapunov functions as having the following properties, which suffice for the Navier--Stokes equations with polynomial $V[\vec{u}]$.
	\begin{definition}\label{def:lyapunov}
		A polynomial functional $V$ is a `Lyapunov function' for the system \cref{eq:NS_primitive} if it satisfies the following properties.\vspace{0.05in}
		\begin{enumerate}
			\item It vanishes for the zero vector field: $V[0] = 0$.\vspace{0.05in}
			\item It is `positive definite': $V[\vec{u}] > 0$ if $\vec{u}\neq 0$.\vspace{0.05in}
            \item Its time derivative is `negative definite': $\dot{V}[\vec{u}]<0$ if $\vec{u}\neq 0$.
		\end{enumerate}
	\end{definition}
    If a Lyapunov function exists for \cref{eq:NS_primitive} at a given Reynolds number, it follows that the system is globally stable at that Reynolds number. The energy method can be seen as the choice $V[\vec{u}]=E[\vec{u}]$, which satisfies the first two criteria of \cref{def:lyapunov} at all $\nRe$ and satisfies the third criterion if and only if $\nRe<\nRe_E$.
% We note that, in more general settings, Lyapunov functions are defined to satisfy additional technical criteria---for instance, the global stability argument fails if $V$ is allowed to decay to zero as $\|\vec{u}\|\to\infty$. These relatively exotic non-convergent behaviors are excluded by considering only polynomial functionals for the Navier--Stokes equations. For more detail on general Lyapunov functionals for general partial differential equations, we direct the reader to the work of \citet{mironchenko_non-coercive_2019}.

	\subsection{A general schematic for quartic Lyapunov functions}
    \label{sec:schematic}
	We propose the following schematic for quartic Lyapunov functions in the general setting of incompressible flows:
\begin{equation}
\label{eq:lyapunov_proposed}
		V[\vec{u}] = E[\vec{u}]^2 + 2E[\vec{u}]\langle\vec{W},\vec{u}\rangle  + Q[\vec{u}].
	\end{equation}
	Here, $E[\vec{u}]=\tfrac{1}{2}\|\vec{u}\|^2$ is the energy, $\vec{W}$ is a chosen incompressible vector field, and $Q[\vec{u}]$ is a reweighted energy (i.e., a positive definite quadratic form) that decays monotonically under the linearized Navier--Stokes equations. The simple family \cref{eq:lyapunov_primitive} given in \cref{sec:intro} is a special case of~\cref{eq:lyapunov_proposed} with
\begin{equation}\vec{W} = \gamma_0\vec{u}_0,\qquad Q[\vec{u}] = \gamma_1E[\vec{u}] + 2\gamma_2\Re\left[\langle\ovl{\vec{u}_1},\vec{u}\rangle\langle\vec{u}_2,\vec{u}\rangle\right].
	\end{equation}
    We study the particular form \cref{eq:lyapunov_primitive} rigorously in later sections; at present, we outline how our general schematic might satisfy the requirements of \cref{def:lyapunov}. The first criterion of \cref{def:lyapunov} is immediate, since each term of \cref{eq:lyapunov_proposed} vanishes when $\vec{u}=0$. The second and third criteria---positivity of $V$ and negativity of $\dot V$---depend on the choice of $\vec W$ and~$Q$.
    
To see how $V$ can be made positive definite, recall Young's inequality:
    \begin{equation}\label{eq:young}
        |ab|\leq \frac{\eps}{2}a^2 + \frac{1}{2\eps}b^2,
    \end{equation}
    which holds for any $a,b\in\RR$ and any $\eps>0$. Applying the Cauchy--Schwarz inequality and then Young's inequality gives
	\begin{equation}|\langle\vec{W},\vec{u}\rangle| \leq \|\vec{W}\|\|\vec{u}\| 
=\frac{\|\vec{W}\|}{\sqrt{2}}\left(\sqrt 2\,\|\vec u\|\right)
\leq \frac{\|\vec{W}\|}{\sqrt{2}}\left(\eps + \eps^{-1}E[\vec{u}]\right)
	\end{equation}
	for any $\eps>0$, so
\begin{equation}
V[\vec u] \ge \left(1-\sqrt 2\,\|\vec W\|\eps^{-1}\right)E[\vec u]^2 + Q[\vec u] - \sqrt 2\,\|\vec W\|\eps E[\vec u].
\end{equation}
For any $\vec W$, we can ensure positivity of $V[\vec u]$ by choosing $\eps > \sqrt{2}\,\|\vec{W}\|$ and then choosing the quadratic form $Q$ so that $Q[\vec u]> \sqrt{2}\,\eps \|\vec{W}\| E[\vec u]$ for all $\vec u$.
	
Ensuring that $\dot V[\vec u]$ is negative definite is more difficult, and accounts for most of our analysis in later sections. To understand heuristically why it is possible, consider how terms of each degree in $\dot V$ arise from the terms of $V$:
	\begin{equation}\label{eq:arrows}
		\begin{minipage}{0.9\textwidth}
			\centering
			\begin{tikzpicture}
				\node [anchor=north] (note) at (0,1.3) {$V = \underbracket{E[\vec{u}]^2} + \underbracket{2E[\vec{u}]\langle\vec{W},\vec{u}\rangle}  + \underbracket{Q[\vec{u}]}$};
				\node [anchor=north] (note2) at (0,-0.75) {$\dot{V} = [\dot{V}]_4 + [\dot{V}]_3 + [\dot{V}]_2$};
				\draw [-latex, ultra thick, black] (-1.5,0.5) to(-0.95,-0.75);
				\node [black] at (-1.7,0.3) {\small\textbf{L}};
				\draw [-latex, ultra thick, blue] (0.3,0.5) to(-0.75,-0.75);
				\node [blue] at (-0.2,0.3) {\small\textbf{N}};
				\draw [-latex, ultra thick, blue] (2.2,0.5) to(1.4,-0.75);
				\node [blue] at (2.35,0.3) {\small\textbf{L}};
				\draw [-latex, ultra thick, dashed, red] (2.05,0.5) to(0.45,-0.75);
				\node [red] at (1.35,0.3) {\small\textbf{N}};
				\draw [-latex, ultra thick, dashed, red] (0.45,0.5) to(0.3,-0.75);
				\node [red] at (0.7,0.3) {\small\textbf{L}};
			\end{tikzpicture}%
		\end{minipage}
	\end{equation}
	%\begin{equation}\label{eq:arrows}
	%    \begin{minipage}{0.9\textwidth}
		%    \centering
		%    \begin{tikzpicture}
			%        \node [anchor=north] (note) at (0,1.3) {$V = \underbracket{E[\vec{u}]^2} + \underbracket{2E[\vec{u}]\langle\vec{W},\vec{u}\rangle}  + \underbracket{Q[\vec{u}]}$};
			%        \node [anchor=north] (note2) at (0,-0.5) {$\dot{V} = (2E\dot{E} + \text{quartic}) + (\text{cubic}) + (\text{quadratic})$};
			%        \draw [-latex, ultra thick, black] (-1.5,0.5) to(-2.0,-0.5);
			%        \node [black] at (-1.9,0.3) {\small\textbf{L}};
			%        \draw [-latex, ultra thick, blue] (0.3,0.5) to(-0.7,-0.5);
			%        \node [blue] at (-0.3,0.3) {\small\textbf{N}};
			%        \draw [-latex, ultra thick, blue] (2.2,0.5) to(2.5,-0.5);
			%        \node [blue] at (2.55,0.3) {\small\textbf{L}};
			%        \draw [-latex, ultra thick, red] (2.05,0.5) to(0.95,-0.5);
			%        \node [red] at (1.45,0.3) {\small\textbf{N}};
			%        \draw [-latex, ultra thick, red] (0.45,0.5) to(0.8,-0.5);
			%        \node [red] at (0.8,0.3) {\small\textbf{L}};
			%\end{tikzpicture}%
			%\end{minipage}
			%\end{equation}
    Here, $[\dot{V}]_4$, $[\dot{V}]_3$, and $[\dot{V}]_2$ respectively denote the quartic, cubic, and quadratic components of the functional $\dot{V}$. Arrows labeled \textbf{N} denote terms arising from nonlinear self-advection of $\vec{u}$, while arrows labeled \textbf{L} denote terms from linear evolution of $\vec u$ via diffusion and linear advection by $\vec{U}$. The quantity $\dot V$ is negative definite only if $[\dot V]_4$ and $[\dot V]_2$ are each negative definite, and only if they together control the sign-indefinite cubic term in the sense that
    \begin{equation}\label{eq:cubic_dominated}
    \left|[\dot{V}]_3\right| < \left|[\dot{V}]_2 + [\dot{V}]_4 \right|
    \end{equation}
    for all admissible $\vec u$. The terms indicated by solid arrows in~\cref{eq:arrows} are those which might be helpful in making $[\dot V]_4$ and $[\dot V]_2$ negative definite when $\nRe>\nRe_E$; the terms indicated by dashed arrows are unhelpful in this regard, and must be small enough that they are dominated as in~\cref{eq:cubic_dominated}.

    Identifying the schematic~\cref{eq:lyapunov_proposed} among more general possibilities has simplified our task substantially. Indeed, it is not very hard to find a quadratic form $Q$ making $[\dot V]_2$ negative definite and a vector field $\mathbf W$ making $[\dot V]_4$ negative definite, as explained shortly. The challenge remaining is to do this in such a way that \cref{eq:cubic_dominated} also holds and is tractable to verify.  For 2-D plane shear flows, we motivate particular choices of $Q$ and $\mathbf W$ in \cref{sec:symmetries}, and verify the corresponding cubic bound~\cref{eq:cubic_dominated} in \cref{sec:symmetries}. First, however, we discuss the quadratic and quartic terms in $\dot{V}$ at a more general level.

    We first consider the quadratic component $[\dot{V}]_2$, which quantifies the evolution of $Q[\vec u]$ under the linearized Navier--Stokes equations. The linearized Navier--Stokes equations must be stable at $\vec{u}=0$ if global stability of the nonlinear equations is to be possible. A consequence of linear stability is that, for any bounded, negative definite quadratic form $M[\vec u]$, there exists a positive definite quadratic form $Q[\vec{u}]$ such that $[\dot{V}]_2=M[\vec{u}]$ \citep{curtain_introduction_1995}. Thus, at any $\nRe$ where the base flow is linearly stable, there are many choices of $Q$ that make $[\dot V]_2$ negative definite.
    
    %is sufficiently small and can be estimated well enough to verify that the full $[\dot V]_3$ term is controlled as in~\cref{eq:cubic_dominated}.
    
    In our ansatz \cref{eq:lyapunov_primitive}, $Q$ is a low-rank perturbation of $E$. Since we focus on $\nRe$ values where $\dot{E}$ is negative definite except along the span of one (complex) velocity mode $\vec{u}_1$, we can construct $Q$ by modifying $E$ only on the subspace spanned by two modes: the energy-unstable mode $\vec{u}_1$, and an energy-stable mode $\vec{u}_2$ of the same wavenumber. By mixing the two modes in $Q$, we enable $\vec{u}_2$ to `absorb' transient growth of $E[\vec{u}_1]$, so that $\dot Q$ can be negative definite for all $\vec u$.
        
	We now consider the quartic component $[\dot V]_4$, which comprises the evolution of $E^2$ and a nonlinear contribution from the cubic term of $V$. The evolution of $E^2$ is simply
    \begin{equation}\label{eq:lieE2}
        \cL \big(E[\vec{u}]^2\big) = 2E[\vec{u}]\dot{E}[\vec{u}].
    \end{equation}
    Avoiding a quintic term in $\dot V$ relies on the quartic term of $V$ being $E^2$ (cf.\ \cref{rem:polynomials}). However, we know that $\dot{E}[\vec{u}]$ (and thus $2E[\vec{u}]\dot{E}[\vec{u}]$) is positive for certain velocity fields when $\nRe>\nRe_E$, meaning that $E[\vec u(\cdot,t)]^2$ can grow transiently.
    % We also note that, even when $\nRe>\nRe_E$, transient energy growth might be quite rare for a given system. For instance, among $10^4$ random initial conditions for 2-D plane Couette flow at $(\nRe,L)=(240,2)$, \citet{fuentes_global_2022} found that only $7$ initial conditions exhibited transient energy growth, and each of them only momentarily.
	The aim of the cubic term in \cref{eq:lyapunov_proposed} is to produce another quartic term in $\dot V$ that, when added to~\cref{eq:lieE2}, restores negativity. The cubic term of $V$ evolves as
    \begin{equation}
        \begin{aligned}
        \label{eq:lie V cubic}
        \cL\big(2E[\vec{u}]\langle \vec{W},\vec{u}\rangle\big) &= 2\dot{E}[\vec{u}]\langle \vec{W},\vec{u}\rangle + 2E[\vec{u}]\cL\big(\langle \vec{W},\vec{u}\rangle\big)\\
        &=2\dot{E}[\vec{u}]\langle \vec{W},\vec{u}\rangle -2E[\vec{u}]\langle\vec{W},(\vec{u}+\vec{U})\cdot\nabla\vec{u} + \vec{u}\cdot\nabla\vec{U} + \nabla p - \nu\nabla^2\vec{u}\rangle \\
        &= 2E[\vec{u}]\langle\vec{u}\cdot\nabla\vec{W}\cdot\vec{u}\rangle + (\text{cubic}),
        \end{aligned}
    \end{equation}
    where the quartic term has been integrated by parts using incompressibility and no-slip (or no-stress) boundary conditions. Combining  \cref{eq:lieE2,eq:lie V cubic} and recalling expression~\cref{eq:edot_primitive} for $\dot E$ gives
    \begin{equation}\label{eq:V4}
        \begin{aligned}
        [\dot{V}]_4 
        % &= 2E[\vec{u}]\dot{E}[\vec{u}] + 2E[\vec{u}]\langle\vec{u},\nabla\vec{W}\cdot\vec{u}\rangle \\
        &= 2E[\vec{u}]\left( \langle\vec{u}\cdot\nabla(\vec{W}-\vec{U})\cdot\vec{u}\rangle - \nu\|\nabla\vec{u}\|^2\right).
        \end{aligned}
    \end{equation}
	In order for the (negative definite) second term in the parentheses to dominate the sign-indefinite first term when $\nRe>\nRe_E$, one must choose $\vec W$ to partially cancel $\vec U$. When we apply the $V$ ansatz \cref{eq:lyapunov_primitive} to particular flows in \cref{sec:symmetries,sec:analysis}, it is convenient for our analysis to simply choose $\mathbf W$ proportional to the first energy eigenmode that has the same wall-normal symmetry as $\vec{U}$. %These are streamwise-constant modes with sinusoidal velocity profiles, whereas the $\mathbf U$ profiles of Couette and Poiseuille flows are linear and quadratic, respectively.

    One can also understand the negativity of $[\dot{V}]_4$ by considering the energetics of the system. Up to a constant, the expression
		\begin{equation}\label{eq:totalenergy}
			E_\mathrm{tot}[\vec{u}] = E[\vec{u}] +\langle\vec{U},\vec{u}\rangle
		\end{equation}
		represents the total kinetic energy of the fluid, including that of the reference state $\vec{U}$. Nonlinear advection can move energy between the two terms in \cref{eq:totalenergy} by tilting the flow toward or away from the laminar shear profile, but cannot modify the total. As such, the time derivative
		\begin{equation}\dot{E}_\mathrm{tot}[\vec{u}] = -\nu\|\nabla\vec{u}\|^2 - \nu\langle\nabla\vec{U},\nabla\vec{u}\rangle
		\end{equation}
        has a negative definite quadratic term.
		
		By choosing $\vec{W}$ such that $\|\vec{W}-\vec{U}\|$ is small, we partially offset the increase of fluctuation energy by explicitly subtracting off (an approximation of) the `interaction energy' between $\vec{U}$ and $\vec{u}$, from which the fluctuation energy is converted via nonlinear advection. This procedure introduces sign-indefinite cubic terms corresponding to viscous generation of vorticity, but retains a negative definite quartic term corresponding to viscous dissipation of the total energy $E_\mathrm{tot}[\vec{u}]$. This physical interpretation suggests that analogous Lyapunov functions might be constructed for systems with non-uniform internal energy---e.g., those governed by the Boussinesq equations \citep{pedlosky_geophysical_2013}---by augmenting the cubic term in \cref{eq:lyapunov_proposed} to account for other energy reservoirs.

		\subsection{A toy model of quartic Lyapunov functions}
        
		A quartic Lyapunov function following our schematic \cref{eq:lyapunov_proposed} can be illustrated with a toy model. Consider the following ordinary differential equation in the three variables $\vec{z}=(z_1,z_2,z_3)$:
		\begin{equation}\label{eq:toymodel}
			\begin{aligned}
				\dot{z}_1 &= -R^{-1}z_1 + z_2 - z_3(z_1+z_2)\\
				\dot{z}_2 &= -R^{-1}z_2 \;\textcolor{white}{+ z_1}\, - z_3(z_1+z_2)\\
				\dot{z}_3 &= -R^{-1}z_3 + z_2 \hspace{0.05em}+ (z_1+z_2)^2,
			\end{aligned}
		\end{equation}
		where $R>0$. This model is meant resemble a three-mode truncation of~\cref{eq:NS_primitive} for shear flows, where $z_3$ is the amplitude of a streamwise-constant mode, $z_1$ and $z_2$ are amplitudes of modes with the same nonzero streamwise wavenumber,  and $R$ is a Reynolds number. \Cref{fig:toymodel} is a bifurcation diagram showing the zero equilibrium as well as a pair of nonzero equilibria that arise in a saddle-node bifurcation at $R_G\approx5.46$. To confirm that this value of $R_G$ is indeed the global stability threshold, we seek Lyapunov functions verifying global stability at all smaller $R$.
        
		\begin{figure}
			\centering
			\includegraphics[scale=0.35,trim={1.5cm 2.2cm 1.5cm 2.2cm},clip]{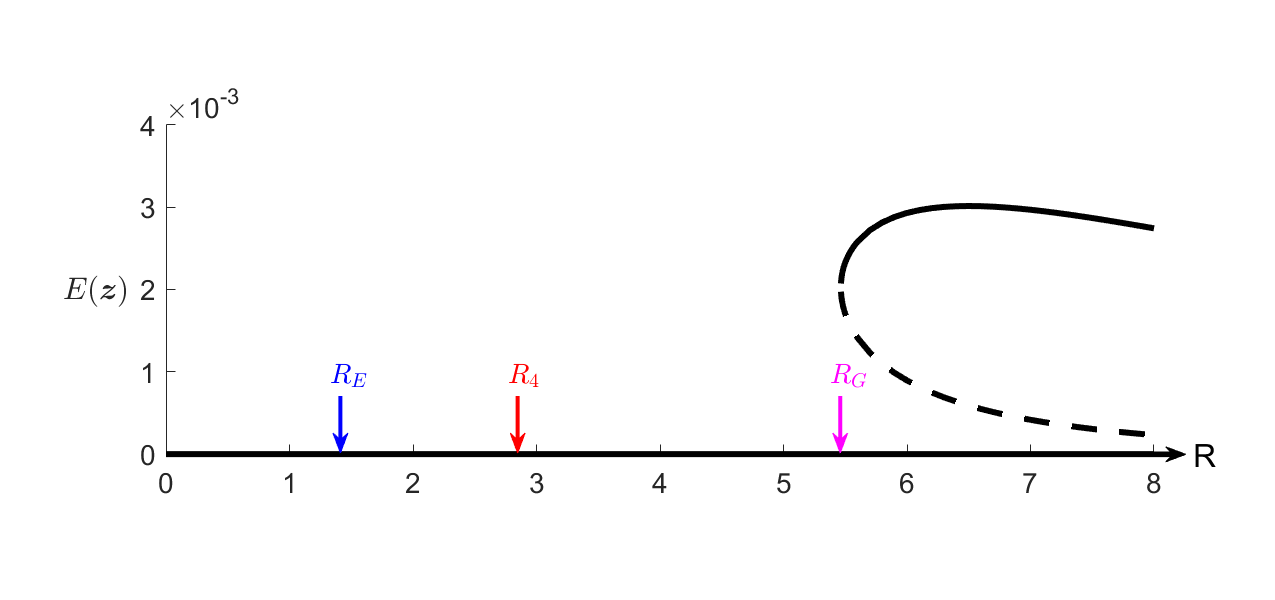}
			\caption{Bifurcation diagram showing stable (solid) and unstable (dashed) equilibria of the toy model \cref{eq:toymodel}. Three values of the parameter $R$ are highlighted: $R_E\approx1.41$ is the largest value at which the (analogue) energy method proves global stability, $R_4\approx 2.85$ is the largest value at which the quartic polynomial \cref{eq:toymodelV} is a Lyapunov function proving global stability, and $R_G\approx5.46$ is the value at which a nonzero equilibrium appears. Global stability for $R$ up to $R_G$ is confirmed by higher-degree polynomial Lyapunov functions (see text). }\label{fig:toymodel}
		\end{figure}
		
We define the energy as $E = \tfrac{1}{2}|\vec{z}|^2$.
As with the Navier--Stokes equations, the time derivative of $E$ is quadratic in $\vec z$  because $E$ is conserved by the nonlinear terms in the model:
		%This quantity is conserved by the `nonlinear convective' terms of \cref{eq:toymodel}, so its time derivative is quadratic:
		\begin{equation}\dot{E} = -2R^{-1}E + z_1z_2 + z_2z_3.
		\end{equation}
Proceeding as in the energy method gives
\begin{equation}
\label{eq:toy Edot}
\dot{E} \leq 2\lambda_1E, \qquad
            \lambda_1=\max_{\vec z\in\mathbb R^3}\frac{\dot E(\vec z)}{2E(\vec z)}=\frac{1}{\sqrt{2}}-\frac{1}{R},
\end{equation}
so energy must decrease monotonically at all $\vec z\neq 0$ when $R<R_E=\sqrt{2}$, implying global stability in this regime. When $R>R_E$, there are some points in $\mathbb R^3$ at which $E(\mathbf{z}(t))$ can increase in time, at least instantaneously. The value of $\lambda_1$ in~\cref{eq:toy Edot} follows from the stationarity condition $\nabla(\dot E/E)=0$.
%with equality when $\vec{x}\propto (1,\sqrt{2},1)^t$. 

To prove global stability at $R$ values larger than $R_E$, our schematic~\cref{eq:lyapunov_proposed} suggests a quartic polynomial of the form
		\begin{equation}\label{eq:toymodelV}
			V = E^2 +\gamma z_3E + \vec{x}\cdot Q\cdot\vec{x}
		\end{equation}
for some real number $\gamma$ and symmetric matrix $Q$. The time derivative $\dot{V}$ can be calculated using the governing model \cref{eq:toymodel}. Both $V$ and $\dot{V}$ are quartic polynomials in $\vec z$ that depend only linearly on $\gamma$ and $Q$. Consequently, for any fixed $R$ value, we can use sum-of-squares (SOS) computational methods to seek $\gamma$ and $Q$ values for which $V(\vec z)$ and $\dot V(\vec z)$ are guaranteed to be positive and negative definite, respectively \citep{parrilo_semidefinite_2003}. This succeeds when $R<R_4\approx 2.85$, with the parameter values
\begin{equation}
\gamma=-1.09, \qquad 
Q =   \begin{pmatrix}
            \phantom{-}0.34& - 0.19 &-0.30\phantom{i}\\
            -0.19 & \phantom{-}1.16 & \phantom{-}0.32\phantom{i}\\-0.30 & \phantom{-}0.32 & \phantom{-}0.73\phantom{i} 
            \end{pmatrix}.
\end{equation}
Our quartic ansatz~\cref{eq:toymodelV} thus verifies global stability up to $R$ roughly twice as large as the analogue energy method. 

In the case of this toy model, one can use SOS computations to construct a sixth-order Lyapunov function $V(\vec{z})$ that verifies global stability up to $R_G=5.46$. In the case of the Navier--Stokes equations, we restrict our analysis to quartic~$V$.

%we can go further to verify stability at $R$ nearly up to $R_G$. Using SOS computations, construct higher-degree polynomial $V(\vec z)$ whose leading terms are higher powers of $E$. For instance, $V$ of degree [what?] verify global stability up to $R=[what?]$. In the case of the Navier--Stokes equations, we restrict our analysis to quartic~$V$.

\section{Complex streamfunctions for two-dimensional shear flows}
\label{sec:preliminaries}
In the remaining sections, we restrict attention to 2-D, wall-bounded shear flows on the dimensionless domain
    \begin{equation}\label{eq:domain}
        \Omega = [-L/2,L/2]\times[-1/2,1/2].
    \end{equation}
		The domain is $L$-periodic in the streamwise coordinate $x$, and velocity fluctuations satisfy the no-slip condition $\vec{u}=0$ on the boundaries at $y=\pm1/2$. \Cref{fig:couette} shows two canonical shear flows: 2-D plane Couette flow, for which the base flow $\vec{U}=y\hat{x}$ is driven by counter-moving walls, and 2-D plane Poiseuille flow, for which $\vec{U} = (4y^2-1)\hat{x}$ is driven by a streamwise pressure gradient. These dimensionless variables fix our convention for the Reynolds number $\nRe$. Specifically, $\nRe=hU/\tilde{\nu}$, where $h$ is the full layer thickness, $U$ is the maximum velocity difference in the laminar flow, and $\tilde{\nu}$ is the dimensional kinematic viscosity. We note that a common alternative convention for Poiseuille flow is to fix $\nRe = hU/2\tilde{\nu}$, corresponding to a length scale of $h/2$.
		
		\begin{figure}
			\centering
			\begin{subfigure}{.46\textwidth}
				\centering
                \hspace*{0.3cm}
				\includegraphics[scale=1.13]{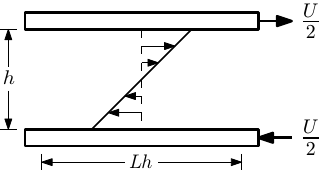}
				\caption{2-D plane Couette flow}
			\end{subfigure}%
			\begin{subfigure}{0.081\textwidth}
				\caption*{$ $}
			\end{subfigure}%
			\begin{subfigure}{.46\textwidth}
				\centering
                \hspace*{0.3cm}
				\includegraphics[scale=1.13]{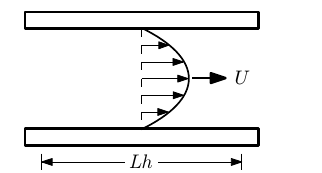}
				\caption{2-D Poiseuille flow}
			\end{subfigure}
			\caption{Schematics of planar Couette and Poiseuille flows in dimensional variables, modified from \citet{fuentes_global_2022} with permission. \textbf{(a)} In Couette flow, the walls move in opposite directions with relative speed $U$ and thus induce the laminar shear profile $\vec{U}=(Uy/h)\,\hat{x}$. \textbf{(b)} In Poiseuille flow, the walls are stationary, but a streamwise pressure gradient induces the laminar flow $\vec{U}=U(4y^2/h^2-1)\,\hat{x}$. In both cases, we scale coordinates such that $h=U=1$.}\label{fig:couette}
		\end{figure}

		We use a streamfunction--vorticity formulation for our analysis of 2-D flows. The scalar streamfunctions $\psi$ and $\Psi$---for the fluctuating and laminar fields, respectively---are related to the velocity fields by
		\begin{equation}
			\vec{u} = \nabla^\perp\psi = (-\psi_y,\psi_x),\qquad \vec{U} = \nabla^\perp\Psi = (-\Psi_y,\Psi_x).
		\end{equation}
		The corresponding scalar vorticity fields are
		\begin{equation}
		\zeta=\nabla^\perp\cdot\vec{u} = \nabla^2\psi,\qquad 
        Z=\nabla^\perp\cdot\vec{U}=\nabla^2\Psi.
		\end{equation}
		In all, the laminar fields for Couette and Poiseuille flows are as follows:
		\begin{equation}
			\begin{aligned}
				\text{\bfseries Couette:} \qquad\qquad & \vec{U}^C = y\hat{x}, &\Psi^C &= -\tfrac{1}{2}y^2,&Z^C &= -1,\vspace{0.05in}\\
				\text{\bfseries Poiseuille:} \qquad\qquad & \vec{U}^P = (1-4y^2)\hat{x},\quad& \Psi^P &= \tfrac{4}{3}y^3 -y,\quad&Z^P&=8y,
			\end{aligned}
		\end{equation}
        where $C$ or $P$ superscripts denote expressions specific to Couette or Poiseuille flow.
		
		To write the Navier--Stokes equations in streamfunction--vorticity formulation, we use the 2-D Jacobian denoted by
		\begin{equation}\label{eq:jacobian}
			\begin{gathered}
				\{a,b\} = a_xb_y - a_yb_x,\\
				\{a,\vec{u}\} = -\{\vec{u},a\} = a_x\vec{u}_y - a_y\vec{u}_x,\\
				\{\vec{u},\vec{v}\} = \vec{u}_x\cdot\vec{v}_y - \vec{u}_y\cdot\vec{v}_x
			\end{gathered}
		\end{equation}
		for any scalar fields $a,b$ and vector fields $\vec{u},\vec{v}$ on $\Omega$. Then, the streamfunction $\psi$ and vorticity $\zeta$ evolve according to
		\begin{equation}\label{eq:NS}
			\zeta_t + \{\Psi + \psi,\zeta\} + \{\psi,Z\}= \nu\nabla^2\zeta,\qquad \zeta = \nabla^2\psi,
		\end{equation}
		with no-slip conditions on $\psi$:
		\begin{equation}\label{eq:BC}\nabla\psi = 0\quad \text{at}\quad y=\pm 1/2.\end{equation}
        Notice that the boundary conditions~\cref{eq:BC} are stronger than enforcing only Neumann conditions on $\psi$ but weaker than enforcing both Neumann and Dirichlet conditions. We next discuss the space of admissible streamfunctions and its subspaces.
		
		\subsection{Function spaces for streamfunctions}
        
		The set of admissible streamfunctions can be written in the following way:
		\begin{equation}\label{eq:cH}
			\cH_\RR = \big\{\phi:\Omega\to\RR\;\big|\;\phi,\nabla\phi\in L^2,\;\nabla\phi|_{\partial\Omega}= 0\big\}.
		\end{equation} 
		We discuss the functional-analytic details of $\cH_\RR$ in \cref{app:space}. 
        
        \begin{remark}
The streamfunction in $\cH_\RR$ only needs one well-defined spatial derivative, even though the 2-D Navier--Stokes equations \cref{eq:NS} appear to call for four in the term $\nabla^2\zeta = \nabla^4\psi$. For any $\psi(\vec{x})\in\cH_\RR$, there is a unique solution $\tilde{\psi}(\vec{x},t)$ to \cref{eq:NS} with initial condition $\tilde{\psi}(\vec{x},0)=\psi(\vec{x})$, and $\tilde{\psi}$ is smooth and well-defined for all $t>0$ \citep{ladyzhenskaia1963}. Physically speaking, any nonzero viscosity is sufficient to immediately smooth an initial condition in $\cH_\RR$.
        \end{remark}

        It is often helpful to decompose the streamfunction into its streamwise Fourier modes:
		\begin{equation}\psi(x,y) = \sum_{k=-\infty}^\infty \hat{\psi}_k(y)e^{2\pi ikx/L}.
		\end{equation}
		Even though the total streamfunction $\psi$ is real-valued, each of its Fourier modes is generally complex-valued. To accommodate this, we complexify $\cH_\RR$ as
		\begin{equation}\label{eq:cH_C}
			\cH_\CC = \{\phi_1 + i\phi_2\;|\;\phi_1,\phi_2\in\cH_\RR\}.
		\end{equation} 
		We also make use of the spaces spanned by each Fourier mode separately,
		\begin{equation}\label{eq:cH_k}
			\begin{gathered}
				%\cH_\CC = \cdots \cH_{-2}+\cH_{-1}+\cH_0+\cH_1+\cH_2\cdots,\\
                \cH_k = \big\{\phi\in\cH_\CC\;\big|\;\phi(x,y) = \hat{\phi}(y)e^{2\pi i kx/L}\big\},
			\end{gathered}
		\end{equation}
        as well as those spanned by modes with only positive or negative values of $k$:
		\begin{equation}
			\cH_+ = \left\{\sum_{k=1}^\infty \phi_k\in\cH_{\mathbb{C}}\;\bigg|\;\phi_k\in\cH_k\right\},\qquad \cH_- = \left\{\sum_{k=1}^\infty \phi_{-k}\in\cH_{\mathbb{C}}\;\bigg|\;\phi_{-k}\in\cH_{-k}\right\}.
		\end{equation}
		In terms of set summation notation, we can decompose $\cH_\CC = \cH_0+\cH_++\cH_-$. 
        
        For complex-valued spaces like $\cH_{\mathbb C}$, it is typical to employ `Hermitian' inner products, in which one argument is conjugated. In the present work, we retain the bilinear form $\langle\,\cdot\,,\cdot\,\rangle$ defined by \cref{eq:innerprod} without complex conjugates. In this notation, the $L^2$ norm of a vector field in $\cH_{\mathbb C}$ is
		\begin{equation}
        \|\vec{u}\|^2 = \langle\ovl{\vec{u}},\vec{u}\rangle = \int |\vec{u}|^2,
        \end{equation}
		and likewise for complex-valued scalar or tensor fields.

		\subsection{Energy eigenmodes in the streamfunction formulation}
		
		We now formulate the energy method and its associated eigenproblem using the streamfunction. The energy of $\vec u=\nabla^\perp\psi$ is
		\begin{equation}
        \label{eq: E psi}
        E[\psi] = \tfrac{1}{2}\|\nabla^\perp\psi\|^2 = \tfrac{1}{2}\|\nabla\psi\|^2,
        \end{equation}
		and its time evolution is
		\begin{equation}\label{eq:edot}
			\dot{E}[\psi] = \langle\nabla\psi,\nabla\psi_t\rangle = -\langle\psi,\zeta_t\rangle = \langle\psi,\{\Psi,\zeta\}\rangle - \nu\|\zeta\|^2.
		\end{equation}
    This expression is quadratic in $\psi$ (since $\zeta=\nabla^2\psi$), just as the expression \cref{eq:edot_primitive} is quadratic in $\vec{u}$. One can write $\dot{E}[\psi]$ in terms of the following self-adjoint operator on $\cH_\RR$:
		\begin{equation}
        \label{eq:energystability}
			\dot{E}[\psi] = \langle\psi, \cE\psi\rangle,\qquad \cE\psi = \tfrac{1}{2}\nabla\cdot(Q_\Psi\cdot\nabla\psi) -\nu \nabla^4\psi,
		\end{equation}
		where we define
		\begin{equation}\label{eq:bigQ}Q_f = \begin{pmatrix} 2f_{xy}& f_{yy}-f_{xx}\\f_{yy}-f_{xx} & -2f_{xy}\end{pmatrix}
		\end{equation}
		for any scalar field $f$. The operator $\cE$ is the `energy stability operator' for the streamfunction formulation. In 2-D Couette and Poiseuille flows, it simplifies as follows:
        \begin{equation}
			\begin{aligned}
				\text{\bfseries Couette:} \qquad\qquad & 
                \cE\psi = -\psi_{xy}-\nu\nabla^4\psi\\
				\text{\bfseries Poiseuille:} \qquad\qquad & 
                \cE\psi = 8y\psi_{xy} + 4\psi_x - \nu\nabla^4\psi.
			\end{aligned}
		\end{equation}
		% \begin{equation}
		% 	\cE_C:\psi\mapsto -\psi_{xy}-\nu\nabla^4\psi,\qquad \cE_P:\psi\mapsto 8y\psi_{xy} + 4\psi_x - \nu\nabla^4\psi,
		% \end{equation}
		% respectively.
        
        In terms of the streamfunction, the Euler--Lagrange equation associated with the maximization that defines $\lambda_1$ in~\cref{eq:edot_to_e} is the following (generalized) eigenproblem:
		\begin{equation}\label{eq:eigenmodes}
		\cE\varphi_j = -\lambda_j\nabla^2\varphi_j.
		\end{equation}
        Note that integrating~\cref{eq:eigenmodes} against $\overline{\varphi}_j$ gives $\dot E[\varphi_j]=2\lambda_j E[\varphi_j]$, in accordance with \cref{eq:edot_to_e}. Since $\cE$ and $\nabla^2$ are both self-adjoint on $\cH_\RR$ with respect to the inner product $\langle\,\cdot\,,\cdot\,\rangle$, the eigenproblem has infinite eigenpairs $(\lambda_j,\varphi_j)$, and all $\lambda_j$ are real. The energy eigenmodes $\varphi_j$ are a complete basis for $\cH_{\mathbb C}$, and they can be chosen to be orthogonal in the sense that
		\begin{equation}
			\langle\nabla\ovl{\varphi_{i}},\nabla\varphi_j\rangle = \delta_{ij},
		\end{equation}
		with $\delta_{ij}$ denoting the Kronecker delta. It will be useful to expand the real streamfunction $\psi\in\cH_{\mathbb R}$ in the complex energy eigenbasis $\varphi_j\in\cH_{\mathbb C}$ as
		\begin{equation}\label{eq:expansion1}
			\psi(x,y) = \sum_j c_j\ovl{\varphi_j}(x,y)/\|\nabla\varphi_j\|^2,\qquad c_j = \langle\nabla\varphi_j,\nabla\psi\rangle = -\langle\varphi_j,\nabla^2\psi\rangle.
		\end{equation}

		\section{Restricting the search for Lyapunov functions}
        \label{sec:symmetries}
        
		In terms of the streamfunction $\psi$ and its energy eigenmodes, we seek a finite set of modes $\varphi_{j_1},...,\varphi_{j_N}$ and a Lyapunov function of the form
		\begin{equation}
        \label{eq:V_revised}
			V[\psi] = E[\psi]^2 + P(c_{j_1},...,c_{j_N},q^2).
        \end{equation}
		Here, $P$ is a real-valued polynomial with cubic and quadratic terms in $N+1$ complex variables, and $q^2$ is the energy of the `tail' of $\psi$, which is the part of $\psi$ orthogonal to $\op{span}\{\varphi_{j_1},...,\varphi_{j_N}\}$.
		
		A primary difficulty with \cref{eq:V_revised} is its potentially high dimension. As $N$ grows large, the space of cubic polynomials has dimension $O(N^3)$, and the space of quartic polynomials (arising in the analysis of $\dot V$) has dimension $O(N^4)$. In the numerical implementation of \citet{fuentes_global_2022} and \citet{iligaray_improved_2026}, the sum-of-squares computations used to verify positivity of $V$ and negativity of $\dot V$ were thus very costly. Their computations reached at most $N=13$ modes, in which case verifying stability at a single $(L,\nRe)$ pair required several hours of CPU time and about 40~GB of memory.
		
		We here reduce the problem size in two steps. In \cref{sec: sym modes,sec: sym V}, we show how the symmetries of Couette or Poiseuille flow can be imposed on $V$ to greatly restrict the ansatz~\cref{eq:V_revised}. In \cref{sec: V simple}, taking inspiration from the physical arguments of \cref{sec:schematic}, we restrict $V$ further to arrive at the simplest possible class of quartic Lyapunov functions for planar shear flows. This ansatz aids our subsequent analysis and---as discussed in \cref{sec:conclusion}---it suggests ways to greatly improve efficiency in future numerical investigations. 
		
		\subsection{Symmetry constraints for the energy basis}
        \label{sec: sym modes}
        
        There remains freedom in the selection of an energy eigenbasis from the eigenproblem~\cref{eq:eigenmodes}. Specifically, every eigenfunction that is not streamwise-constant belongs to a two-dimensional eigenspace, from which one can choose any two orthogonal functions. \citet{fuentes_global_2022} and \citet{iligaray_improved_2026} chose all modes to be real-valued, so that the coefficients $c_j$ appearing in the expansion \cref{eq:expansion1} would themselves be real-valued. Here, we instead choose complex eigenmodes, so that they can simultaneously be eigenmodes of the models' symmetry transformations. This choice greatly simplifies our description of symmetry-invariant Lyapunov functions in the next subsection.
		
		%\begin{lemma}\label{lem:symmetry}
		%    Suppose our fluid flow is invariant under the action of a compact symmetry group $G$, and suppose a Lyapunov function $V$ exists in the sense of \cref{def:lyapunov}. Then a Lyapunov function $V_0$ exists that is itself invariant under $G$, i.e., it satisfies $V_0\circ g = V_0$ for any $g\in G$.
		%\end{lemma}
		%The lemma follows from considering the properties of the averaged functional
		%\[V_0[\psi] = \int_G V[g\cdot\psi]\,d\mu(g),\]
		%where $\mu$ is the unique normalized Haar measure on $G$ \citep{folland_course_1994}. It is proved in a finite-dimensional setting, for instance, in the thesis of \citet{oeri_convex_2023}, but the infinite-dimensional proof is straightforward. 
		
		All plane parallel shear flows are invariant under streamwise translations by any $x_0\in\RR$,
		\begin{equation}
        \label{eq:translation}
			T_{x_0}:\varphi(x,y)\mapsto \varphi(x-x_0,y).
		\end{equation}
        On top of this, Couette and Poiseuille flows each have a discrete symmetry. In Couette flow, the base velocity field $\mathbf U$ is odd in the wall-normal coordinate $y$, so the streamfunction dynamics~\cref{eq:NS} are invariant under a $180^\circ$ rotation in the $x$--$y$ plane. For convenience, we define this rotation operator with a complex conjugate as well,
		\begin{equation}
			\cR^C:\varphi(x,y)\mapsto \ovl{\varphi(-x,-y)},
		\end{equation}
		so that $\cR^C$ maps each single-wavenumber space $\cH_k$ to itself. In Poiseuille flow, $\mathbf U$ is even in $y$, so the streamfunction dynamics are invariant under a vertical reflection,
		\begin{equation}
			\cR^P:\varphi(x,y)\mapsto -\varphi(x,-y).
		\end{equation}
		
        Translation-invariant shear flows have energy eigenmodes that are sinusoidal (or constant) in the streamwise direction. They can be chosen to take the form of complex plane waves,
		\begin{equation}
        \label{eq:eigenmodes_waves}
			\varphi_j(x,y) = \hat{\varphi}_j(y) e^{2\pi i k_j x/L} \in \cH_{k_j},
		\end{equation}
        where $j$ indexes a chosen ordering of energy eigenmodes, and $k_j\in\ZZ$ are the corresponding streamwise wavenumbers. Energy eigenmodes of the form~\cref{eq:eigenmodes_waves} are also eigenmodes of $T_{x_0}$, since 
		\begin{equation}T_{x_0}\varphi_j = e^{-2\pi i k_j x_0/L}\varphi_j.\end{equation}
        We note that it is common to label 2-D eigenmodes by a double-indexing scheme $(k,m)$, as in \citet{fuentes_global_2022}, where $k$ is the streamwise wavenumber and $m$ is a wall-normal mode number. A single-indexing scheme is more convenient here, since our analysis in the next section uses only three complex modes.

        After choosing an ordering for the $j$ indexing, it remains only to normalize the modes and choose their complex phases. We normalize as follows:
		\begin{equation}\label{eq:eigenmodes_norm}
			\begin{aligned}
				\|\nabla\varphi_j\|^2 &= 1/2\qquad &\text{if}\;\varphi_j\in\cH_0,\\
				\|\nabla\varphi_j\|^2 &= 1\qquad&\text{if}\; \varphi_j\notin\cH_0.
			\end{aligned}
		\end{equation}
		To choose the phases, we observe from linear-algebraic considerations that the eigenmodes $\varphi_j$ of $\cE$ will also be eigenmodes of $\cR^{C}$ in Couette flow and of $\cR^P$ in Poiseuille flow. Since both maps square to the identity, their only possible eigenvalues are $\pm 1$, and we find
		\begin{equation}\label{eq:eigenvalues_R}
			\begin{gathered}
				\hat{\varphi}^C_j(y)e^{2\pi ik_jx/L} = \pm\cR_C \left[\hat{\varphi}^C_j(y)e^{2\pi ik_jx/L}\right] = \pm \ovl{\hat{\varphi}^C_j(-y)}e^{2\pi ik_jx/L},\\
				\hat{\varphi}^P_j(y)e^{2\pi ik_jx/L} = \pm\cR_P \left[\hat{\varphi}^P_j(y)e^{2\pi ik_jx/L}\right] = \mp \hat{\varphi}^P_j(-y)e^{2\pi ik_jx/L}. 
			\end{gathered}
		\end{equation}
        We fix the phases of all Couette modes $\varphi_j^C$ by requiring them to satisfy the following `Hermitian' property:
		\begin{equation}
        \label{eq:eigenmodes_herm}
			\hat{\varphi}^C_j(y) = \ovl{\hat{\varphi}^C_j(-y)}.
        \end{equation}
        Note that not all streamwise-constant modes are real-valued with this convention; streamwise-constant modes that are even or odd in $y$ are purely real or imaginary, respectively. We fix the phase of Poiseuille modes $\varphi_j^P$ that are streamwise-constant by requiring them to be real-valued,
        \begin{equation}
        \label{eq:eigenmodes_real}
			\varphi_j^P\in\cH_\RR\qquad \text{if}\;\varphi_j^P\in\cH_0.
		\end{equation}
        We do not choose phases at this stage for the $\varphi_j^P$ modes with nonzero $k_j$, but we revisit the matter in \cref{sec:analysis}.
		
		\subsection{Symmetry constraints for the Lyapunov function}
        \label{sec: sym V}
        
		A classic result in dynamical systems is that, if a Lyapunov function exists for a given system, then one exists that is invariant under all (compact) symmetries of the system. To see why, consider the group of streamwise translations \cref{eq:translation}, which form a compact group because the $x$ direction is periodic. Suppose $\tilde{V}[\psi]$ is a Lyapunov function satisfying the three conditions of \cref{def:lyapunov} for a given 2-D shear flow at a fixed Reynolds number. Using invariance of the system under $T_{x_0}$, one can show that $\tilde{V}[T_{x_0}\psi]$ satisfies the same three conditions for any $x_0$, and so does its average over the symmetry group,
		\begin{equation}
			V[\psi] = \frac{1}{L}\int_{-L/2}^{L/2} \tilde{V}[T_{x_0}\psi]\,\d x_0.
		\end{equation}
This latter Lyapunov function is invariant under streamwise translations. Therefore, in the search for Lyapunov functions, one can restrict to translation-invariant functionals at the outset. More general versions of this argument can be found in the literature \citep[e.g.,][]{oeri_convex_2023}.
		
Symmetry lets us greatly restrict the $V$ ansatz \cref{eq:V_revised} for 2-D shear flows. Namely, for any fixed set of energy modes $\varphi_{j_1},...,\varphi_{j_N}$, one only needs to consider the relatively small subspace of polynomials $P(c_{j_1},...,c_{j_N},q^2)$ that are invariant under translation and discrete symmetries (if there are any). We enumerate these symmetry-invariant polynomials in the following lemma:
%\Cref{lem:symmetry_specific} gives the most general form of such quartic $V$ as a function of $(c_{j_1},...,c_{j_N},E)$, where we note that dependence on $E$ or on $q^2$ are interchangeable since $E -q^2$ can be written as a quadratic function of the $c_j$. The same result could be stated using only real variables, but it is far simpler and easier to implement in complex variables.

		%Below, we say that a functional $V[\psi]$ is a `polynomial' in $\psi$ if it can be decomposed as
		%\begin{equation}
			%V[\psi] = V_0 + V_1[\psi] + V_2[\psi,\psi] + \cdots V_M[\psi,...,\psi]
		%\end{equation}
		%for some $M\geq 0$, where $V_m$ is $m$-multilinear in each of its arguments; we say that $V$ is `quartic', `cubic', or `quadratic' if it is a polynomial with $M=4$, $3$, or $2$, respectively. We recover the following lemma:
		\begin{lemma}\label{lem:symmetry_specific}
			Fix a 2-D shear flow with streamwise period $L$, Reynolds number $\nRe$, and streamwise-constant laminar state $\Psi(y)$. Fix a finite set of energy modes $\varphi_j\in\cH_\CC$ with corresponding streamwise wavenumbers $k_j\in\ZZ$, as in \cref{eq:eigenmodes_waves}. Assume the mode set to be closed under complex conjugation, denoting $\varphi_{-j} = \ovl{\varphi_j}$. Suppose there exists a quartic Lyapunov function $\tilde{V}[\psi]$ of the form \cref{eq:V_revised} that is non-degenerate in the sense that $\tilde{V}[\psi]\geq \eps\|\nabla\psi\|^2$ for some $\eps>0$. Then there exists a translation-invariant quartic Lyapunov function $V[\psi]$ of the form
			\begin{equation}\label{eq:lemma_ansatz}
				V[\psi] = E[\psi]^2 + \sum_{\substack{i\\k_i=0}}b_ic_iE[\psi] + \sum_{\substack{i,j,\ell\\k_i+k_j+k_\ell=0}}B_{ij\ell}c_{i}c_{j}c_{\ell} + \sum_{\substack{i,j,k\\ k_i,k_j=k} }A_{k;ij}\ovl{c_{i}}c_{j} + aE[\psi],
			\end{equation}
where $c_j = \langle\nabla\varphi_j,\nabla\psi\rangle$, for some values of $a>0$ and $b_i,A_{k;ij},B_{ij\ell}\in\CC$. For Couette or Poiseuille flow, there exists a symmetry-invariant $V$ where~\cref{eq:lemma_ansatz} is further constrained as follows:\vspace{0.05in}
            \begin{enumerate}
                \item For Couette flow, suppose phases of $\varphi_j$ are chosen such that \cref{eq:eigenmodes_herm} holds. Then, all coefficients are real-valued, and, for indices corresponding to modes that are streamwise-constant (i.e., wavenumber-zero), $b_i\neq 0$ only if $\varphi_i\in\cH_0$ is even in $y$, $A_{0;ij}\neq 0$ only if $\varphi_i\varphi_j\in\cH_0$ is even in $y$, and $B_{ij\ell}\neq 0$ only if $\varphi_{i}\varphi_{j}\varphi_{\ell}\in\cH_0$ is even in $y$.\vspace{0.05in}
                \item For 2-D Poiseuille flow, suppose phases of streamwise-constant $\varphi_j$ are chosen such that \cref{eq:eigenmodes_real} holds. Then, the $b_i$ are all real-valued, and $b_i\neq 0$ only if $\varphi_i\in\cH_0$ is odd in $y$, $A_{k;ij}\neq 0$ only if $\varphi_i\varphi_j$ is odd in $y$, and $B_{ij\ell}\neq 0$ only if $\varphi_{i}\varphi_j\varphi_\ell$ is odd in $y$.
            \end{enumerate}
             
		\end{lemma}

        The justification of \cref{lem:symmetry_specific} is fairly straightforward. For one, it is clear that we must have $a>0$ in all cases, to ensure that $\tilde{V}[\psi]\geq \eps\|\nabla\psi\|^2$ when $\psi$ is small. The wavenumber coupling restrictions in \cref{eq:lemma_ansatz} are necessary to ensure that $V[T_{x_0}\psi] = V[\psi]$ for any translation $T_{x_0}$. For instance, if the term $c_ic_jc_\ell$ were such that $k_i+k_j+k_\ell=k\neq 0$, then $T_{x_0}$ would map
        \begin{equation}
            T_{x_0}:c_ic_jc_\ell\mapsto e^{2\pi i kx_0/L}c_ic_jc_\ell
        \end{equation}
and thus modify the value of $V[\psi]$. The additional claims for Couette and Poiseuille flow follow likewise from the actions of $\cR^C$ and $\cR^P$, respectively. Indeed, \cref{eq:eigenmodes_herm} implies
        \begin{equation}
            \cR^C:c_i\mapsto \ovl{c_i},\qquad \cR^P:c_i\mapsto \pm c_i,
        \end{equation}
        with the latter sign depending on the $y$-parity of $\varphi_i$. Requiring the full polynomial \cref{eq:lemma_ansatz} to remain invariant under each of these maps yields the stated constraints on coefficients.
		
		\Cref{lem:symmetry_specific} allows for significant dimension reduction in the polynomial space from which we seek $V$. For instance, the smallest real mode set used by \citet{fuentes_global_2022} for 2-D Couette flow has six modes. If they had not exploited symmetry (which they partially did), the quadratic term of their $V$ would be parameterized by a $6\times 6$ symmetric matrix and a coefficient on $E$, totaling 22 independent coefficients. On the other hand, \cref{lem:symmetry_specific} justifies a reduction to 6 coefficients, as we will see in the next subsection. Symmetry reduction in the cubic term is even more striking, reducing 62 real coefficients to 6, and the total reduction in coefficients becomes increasingly significant as the number of modes $N$ is raised. %Finally, and most importantly for computational purposes, symmetry reduction leads to structure in the SOS constraints that must be imposed in the final computational step of \cref{sec:finalsos}.

        \subsection{Minimal quartic Lyapunov functions for shear flows}
        \label{sec: V simple} 

        With \cref{lem:symmetry_specific} in place, we can now present our minimal Lyapunov function ansatzes for 2-D Couette and Poiseuille flow. These ansatzes are specific to a regime of $\nRe>\nRe_E$ where there is only one unstable energy eigenvalue $\lambda_1>0$, but may inform the construction of quartic Lyapunov functions more generally. 
        
        We first consider Couette flow, where we can take direct inspiration from the work of \citet{fuentes_global_2022}. The smallest mode set they consider contains six modes, which we adapt to our complex notation as $(\varphi_0^C,\varphi_{\pm 1}^C,\varphi_{\pm 2}^C,\varphi_3^C)$. Specifically:\vspace{0.05in}
		\begin{enumerate}
			\item The mode $\varphi_0^C\in\cH_0$ is $\varphi_0^C(y)=\frac{1}{2\pi\sqrt{L}}\cos(2\pi y)$. It is even in $y$ and thus purely real by~\cref{eq:eigenmodes_herm}.\vspace{0.05in}
			\item The modes $\varphi_{\pm 1}^C\in\cH_1$ are the only energy-unstable modes. These have wavenumber $k=\pm1$ and wall-normal mode number $m=1$. %Its real and imaginary parts are $\pi/2$ phase shifts of each other.
            \vspace{0.05in}
			\item The modes $\varphi_{\pm 2}^C\in\cH_2$ have wavenumber $k=\pm 1$ and wall-normal mode number $m=2$. %Its real and imaginary parts are $\pi/2$ phase shifts of each other.\vspace{0.05in}
			\item The mode $\varphi_3^C\in\cH_0$ is $\varphi_3(y)=\frac{i}{\pi\sqrt{L}}\sin(\pi y)$. It is odd in $y$ and thus purely imaginary by~\cref{eq:eigenmodes_herm}.\vspace{0.05in}
		\end{enumerate}
        Recall that we denote $\varphi_{-j}=\ovl{\varphi_j}\in\cH_{-j}$ and $c_{-j} = \ovl{c}_j$ for complex modes. The real modes denoted as $\mathbf u_n$ in \citet{fuentes_global_2022} relate to ours by $\mathbf u_n=\nabla^\perp\eta_n$, where
		\begin{equation}
        \label{eq:realcorrespondence}
			\begin{gathered}
				\eta_1=-i\varphi_3^C,\qquad \eta_2=\varphi_0^C,\vspace{0.05in}\\
				\eta_3=\Re\varphi_1^C,\qquad \eta_4 = \Im\varphi_1^C,\qquad \eta_5=\Re\varphi_2^C,\qquad \eta_6=\Im\varphi_2^C,
			\end{gathered}
		\end{equation}
        recalling that $\mathfrak R$ and $\mathfrak I$ denote real and imaginary parts of complex numbers, respectively. The top row of \cref{fig:energymodes} depicts the modes $\varphi_0^C$, $\varphi_1^C$ and $\varphi_2^C$ for Couette flow. In what follows we do not use $\varphi_3^C$, but it coincides with the mode $\varphi_0^P$ we will use for Poiseuille flow, shown in the bottom row of \cref{fig:energymodes}.
        
		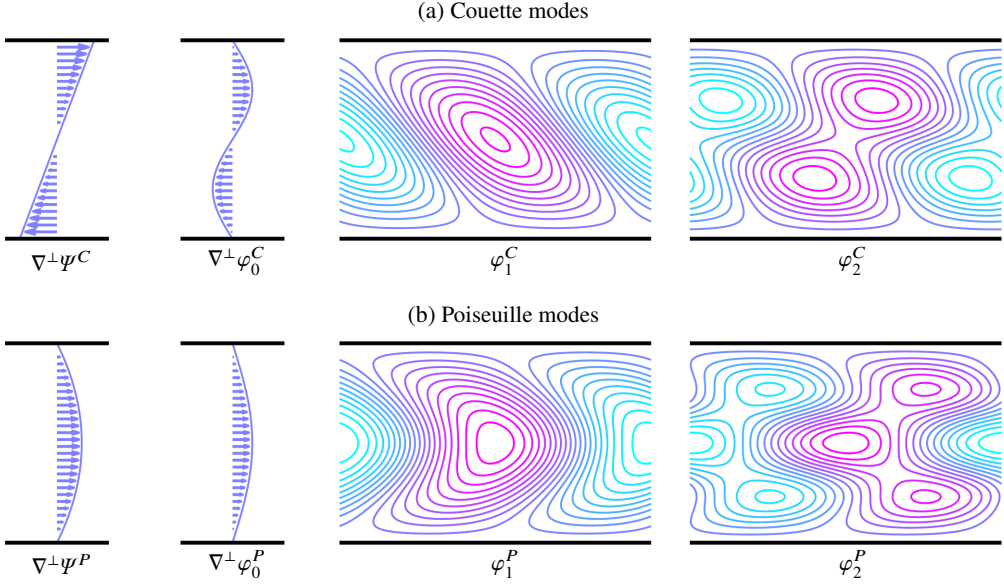
\begin{figure}
			\centering
            \begin{subfigure}{\textwidth}
            \centering
            \caption{Couette modes}
                \begin{tikzpicture}
				%\node (img) {\includegraphics[scale=0.16,trim={8cm 0 0 0}, clip]{presentation/new_modes.png}};
                \node (img) {\includegraphics[scale=0.24]{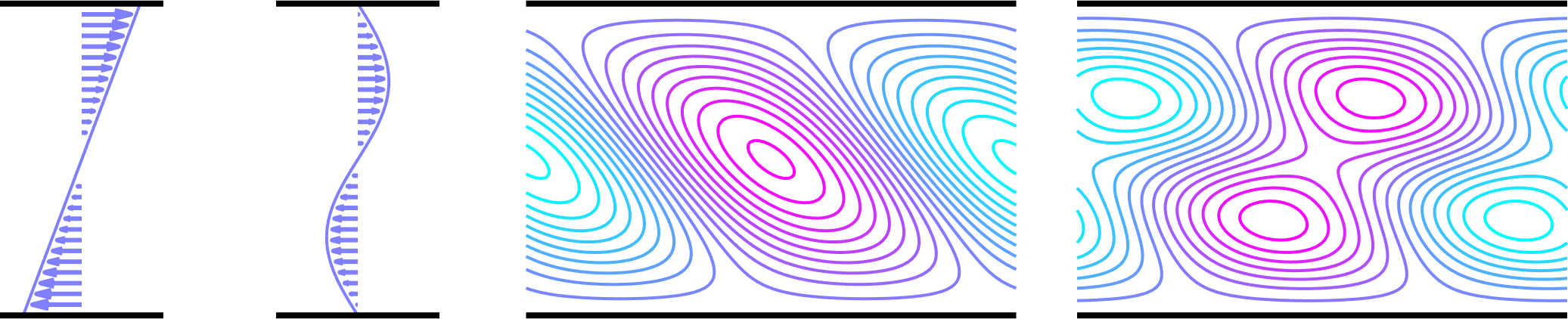}};
				%\draw (-5.4,-1.5) node [text=black,left] {\footnotesize $\varphi_3=\eta_1$};
                \draw (-5.3,-1.6) node [text=black,left] {\footnotesize $\nabla^\perp\Psi^C$};
				\draw (-3.0,-1.6) node [text=black,left] {\footnotesize $\nabla^\perp\varphi_0^C$};
				\draw (0.35,-1.6) node [text=black,left] {\footnotesize $\varphi_1^C$};
				\draw (4.95,-1.6) node [text=black,left] {\footnotesize $\varphi_2^C$};
			\end{tikzpicture}
            \end{subfigure}
            \begin{subfigure}{\textwidth}
            \centering
            \caption{Poiseuille modes}
                \begin{tikzpicture}
				%\node (img) {\includegraphics[scale=0.16,trim={8cm 0 0 0}, clip]{presentation/new_modes.png}};
                \node (img) {\includegraphics[scale=0.24]{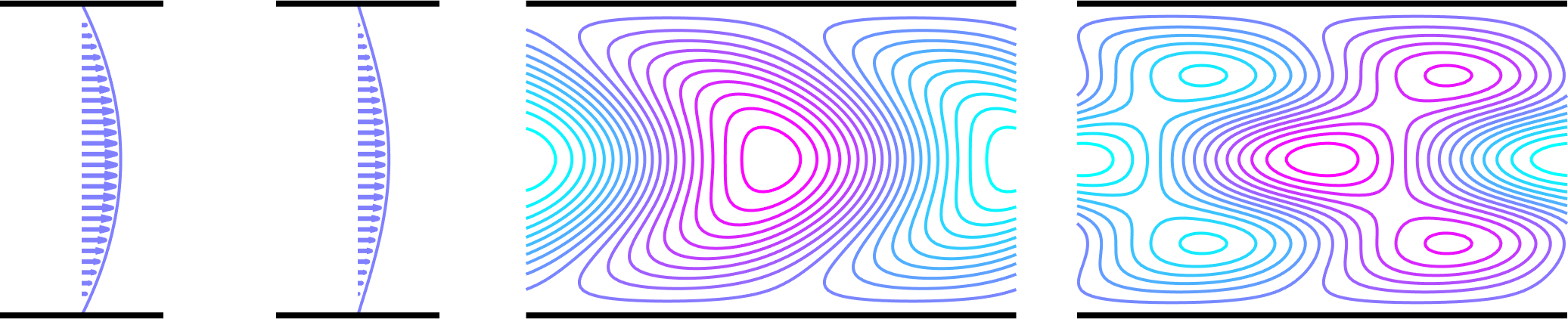}};
				%\draw (-5.4,-1.5) node [text=black,left] {\footnotesize $\varphi_3=\eta_1$};
                \draw (-5.3,-1.6) node [text=black,left] {\footnotesize $\nabla^\perp\Psi^P$};
				\draw (-3.0,-1.6) node [text=black,left] {\footnotesize $\nabla^\perp\varphi_0^P$};
				\draw (0.35,-1.6) node [text=black,left] {\footnotesize $\varphi_1^P$};
				\draw (4.95,-1.6) node [text=black,left] {\footnotesize $\varphi_2^P$};
			\end{tikzpicture}
            \end{subfigure}
			\caption{For 2-D plane Couette flow (top) and Poiseuille flow (bottom): the laminar base state $\Psi$, and the three energy eigenmodes $\varphi_j$ used in our ansatz. Velocity profiles are shown for the streamwise-constant fields $\Psi$ and $\varphi_0$. Streamlines are shown for the real parts of $\varphi_{1}$ and $\varphi_2$; their real and imaginary parts differ only by streamwise translations of $L/2$. The Couette modes correspond to a subset of the real energy modes of \citet{fuentes_global_2022} via \cref{eq:realcorrespondence}.}
			\label{fig:energymodes}
		\end{figure}		
		
		According to \cref{lem:symmetry_specific}, any quartic Lyapunov function depending only on $E$ and the six chosen modes can be made invariant under the symmetries of Couette flow by restricting to the form
		\begin{equation}\label{eq:ansatz2}
        \begin{aligned}
			V &= E^2 + c_0(w_1E + w_2c_0^2 + w_3c_3^2+B_{ij}\ovl{c_i}c_j) + c_3 C_{ij}\ovl{c_i}c_j\\
            &\qquad \qquad+ (w_4E + w_5c_0^2 + w_6c_3^2 + A_{ij}\ovl{c_i}c_j),
        \end{aligned}
		\end{equation}
        %\begin{equation}\label{eq:ansatz2}
        %\begin{aligned}
		%	V &= E^2 + b_0c_0E + B_{000}c_0^3 + B_{033}c_0c_3^2+B_{0ij} c_0\ovl{c_i}c_j +  B_{3ij}c_3\ovl{c_i}c_j\\
        %    &\qquad \qquad+ aE + A_{0;00}c_0^2 + A_{0;33}c_3^2 + A_{1;ij}\ovl{c_i}c_j,
        %\end{aligned}
		%\end{equation}
		where all coefficients are real, $A$ and $B$ are $2\times 2$ symmetric matrices, $C$ is a $2\times 2$ antisymmetric matrix (recalling that $c_3\in i\,\RR$), and repeated indices imply summation over $i,j\in\{1,2\}$.

To arrive at the simplest possible quartic $V$, we eliminate as many terms as possible from \cref{eq:ansatz2}. First, we retain only terms that fall within the general schematic \cref{eq:lyapunov_proposed} proposed in \cref{sec:schematic}. Only the first cubic term is of this form, so we discard the terms corresponding to $w_2$, $w_3$, $B_{ij}$, and $C_{ij}$. Next, we note that the quadratic term $[\dot{V}]_2$ cannot be negative definite when $\nRe>\nRe_E$ unless $V$ has at least one off-diagonal quadratic term, so the off-diagonal $A_{ij}$ terms are required. It can be shown that the $w_5$ and $w_6$ terms cannot help make $[\dot{V}]_2$ negative definite, so we discard them. Finally, by numerical experimentation with a code similar to that of \citet{fuentes_global_2022}, we find no benefit from diagonal terms in $A_{ij}$, so we set them to zero. Renaming the retained coefficients, we arrive at our chosen minimal ansatz for 2-D shear flows:
		\begin{equation}\label{eq:V}
			V[\psi] = E^2 + 2\gamma_0c_0E + \gamma_1E + \gamma_2 A_{ij}\ovl{c_{i}}c_{j},\qquad A = \begin{pmatrix}
				&e^{+i\theta}\\e^{-i\theta}&
			\end{pmatrix},
		\end{equation}
		where $\gamma_0,\gamma_1,\gamma_2,\theta\in\RR$ are coefficients to be chosen, and summation over the repeated indices $i,j\in\{1,2\}$ is implied. For Couette flow, the matrix $A$ is real per \cref{lem:symmetry_specific}; without loss of generality, we fix $\theta=0$ so that the entries of $A$ are 0 on the diagonal and 1 off the diagonal. However, one can allow for complex $A$ in other shear flows.
        
For Poiseuille flow, we adopt the same ansatz \cref{eq:V}, except that $\varphi_2$ is the energy mode with streamwise wavenumber $k=1$ and wall-normal mode number $m=3$, rather than $m=2$. This change is dictated by symmetry, following \cref{lem:symmetry_specific}, so that $\varphi_1\varphi_2$ is even in $y$. Depending on the relative phase between $\varphi_1$ and $\varphi_2$, the optimal choice of $A$ may not be real-valued. Once an optimal $A$ is known, this phase can be adjusted \emph{a posteriori} so that the optimal $A$ is the same real matrix (with $\theta=0$) as in Couette flow. In~\cref{sec:finalsos}, we give a heuristic for choosing the relative phase between $\varphi_1$ and $\varphi_2$ that allows us to fix $\theta=0$ \emph{a priori} in Poiseuille flow. % \cref{eq:alpha}.

%We remark that---in both cases---the simplifications we have made ensure that our functional conforms to the general schematic \cref{eq:lyapunov_proposed} proposed in \cref{sec:interpretation}. Specifically, the expression \cref{eq:V} is simply a restatement of the functional \cref{eq:lyapunov_primitive} in our refined notation.
        
%We emphasize that, in the case of Couette flow, one knows \emph{a priori} that $\vec{A}$ is real-valued, and thus (without loss of generality) that $\theta=0$. For Poiseuille flow, one can fix $\theta=0$ \emph{a posteriori} by adjusting the relative phase between the complex energy modes $\varphi_1$ and $\varphi_2$. We will discuss a heuristic that allows \emph{a priori} fixing $\theta=0$ for Poiseuille flow in \cref{sec:interpretation_primary} (see \cref{eq:alpha}). In any case, the expression \cref{eq:V} is simply a restatement of the functional \cref{eq:lyapunov_primitive} in our refined notation.

\section{Proving Lyapunov's criteria}
\label{sec:analysis}
        
Having identified the simplest ansatz \cref{eq:V} for $V$, we aim to show that it satisfies the criteria of \cref{def:lyapunov} at chosen $(L,\nRe)$ values, and thus deduce global stability. Most of our analysis is carried out by hand, but two steps are numerical. The first numerical step is computing approximate energy eigenmodes $\varphi_j$ and eigenvalues $\lambda_j$, as well as various integrals of the eigenmodes. The second numerical step is the solution of a small sum-of-squares (SOS) problem, whose computational cost is negligible compared to even the least expensive SOS computations of \citet{fuentes_global_2022} and \citet{iligaray_improved_2026}. We note that our results are not yet rigorous to the standard of a computer-assisted proof, due to these two numerical steps; both could be verified with the aid of interval arithmetic \citep{Kearfott1996}, but we leave this for future work.
		
There are three criteria in~\cref{def:lyapunov}. Our ansatz clearly satisfies the first condition, that $V[0] = 0$. Moreover, it turns out that one does not have to independently verify the second condition---that $V$ is positive definite---because this is implied by the third condition when $V$ is of the chosen form (see \cref{rem:negdef} below). All of our subsequent efforts are to verify the third condition, that $\dot V$ is negative definite.
\begin{remark}
\label{rem:negdef}
The negative definiteness of $\dot{V}$ is sufficient to deduce the positive definiteness of $V$. Indeed, since $V[\psi]\sim E[\psi]^2$ for large arguments, and $V[\psi]=V(c_j,q_n)$ is a polynomial in finitely many variables, it must have at least one finite global minimum. However, one must have $\dot{V}[\psi]=0$ at any global minimum of $V$. If $\dot{V}$ is negative definite (and thus nonzero whenever $\psi\neq 0$), it follows that the unique global minimum of $V$ is at $\psi=0$, and in turn that $V$ is positive definite.
\end{remark}
Taking the time derivative of $V$ gives
\begin{equation}
\label{eq:dotV}
			\dot{V}[\psi] = \dot{E}\left(2E + 2\gamma_0c_0 + \gamma_1\right) + 2\gamma_0E\dot{c}_0 + 2\gamma_2\Re [A_{ij}\ovl{c_i}\dot{c}_j].
\end{equation}
The primary difficulty in showing that $\dot V$ is negative definite is that, unlike $V$, it is not a function on a finite-dimensional space. Indeed, both $\dot{E}$ and $\dot{c}_j$ depend on the entire $\psi$ field, and neither can be determined by only the variables $(c_j,q_n)$. The next two subsections develop a bound of the form \cref{eq:intro_polybound} to reduce the problem to a finite-dimensional setting. We then show how to prove global stability over a continuum of $\nRe$ values, we formulate the necessary SOS problem, and we present concrete global stability results.	
		
\subsection{The time derivative of the Lyapunov function}

		The goal of the present subsection is to write $\dot{V}$ in a form more amenable to further analysis, so that we might show its negative definiteness in what follows. For this, it is helpful to further decompose the streamfunction $\psi$. Let $\cH_*\subset\cH_\CC$ be the complex span of $\varphi_{\pm 1}$ and $\varphi_{\pm 2}$ together. We split the streamfunction into a `primary' component $\psi^a$ and a `tail' component $\psi^b$ as follows:
		\begin{equation}
        \label{eq:psia}
			\psi = \psi^a + \psi^b,\qquad \psi^a(x,y)=2c_0\varphi_0(y) + \sum_{j=\pm 1,\pm2}c_j\ovl{\varphi_j}(x,y).
		\end{equation}
		The energy of the tail is denoted by $q^2=E[\psi^b]$. We decompose the tail into three real components based on the streamwise wavenumber:
		\begin{equation}
        \label{eq:tailcomponents}
			\begin{gathered}
				\psi^b = \psi^b_0(y,t) + \psi^b_1(x,y,t)+\psi^b_2(x,y,t),\vspace{0.05in}\\
				\psi^b_0\in\cH_0\cap\{\varphi_0\}^\perp,\qquad \psi^b_1\in (\cH_1+\cH_{-1})\cap\cH_*^\perp,\qquad\psi^b_2\in\cH_{\pm 2} + \cH_{\pm 3} + \cdots,
			\end{gathered}
		\end{equation}
		where $^\perp$ denotes the orthogonal complement. All components in~\cref{eq:tailcomponents} are orthogonal to $\varphi_0$ and $\cH_*$, which span the possible values of $\psi^a$. When needed, we further split each tail component into positive and negative frequencies:
		\begin{equation}
			\psi^b_n = \psi^b_{n,+} + \psi^b_{n,-},\qquad \psi^b_{n,\pm}\in \cH_{\pm},\qquad n = 1,2.
		\end{equation}
		
		\begin{remark}\label{rem:tailsplitting}
			Splitting the tail into Fourier components allows more sign-indefinite terms in $\dot{V}[\psi]$ to be deleted \emph{a priori} using orthogonality. Similar decomposition can be used to refine the numerical formulation of \citet{fuentes_global_2022}. We have carried out numerical experiments with such refinements, and we indeed find that they can significantly strengthen results.
		\end{remark}
		
		To make the $\dot V$ expression~\cref{eq:dotV} explicit, we need expressions for $\dot{E}$ and $\dot{c}_j$. With the tail decompositions we have introduced, the energy can be written as
		\begin{equation}
        \label{eq:E decomp}
			E[\psi] = \sum_{j=0}^2|c_j|^2 + \sum_{n=0}^2 q_n^2,\qquad q_n^2 = \tfrac{1}{2}\|\nabla\psi^b_n\|^2.
		\end{equation}
		Note our convention that $q^2$ denotes tail energy as in \citet{goulart_global_2012}, rather than $\tfrac12q^2$ as in \citet{fuentes_global_2022} and \citet{iligaray_improved_2026}. For $n=1,2$, we also note that $q_n^2 = \|\nabla\psi^b_{n,\pm}\|^2$. The expression~\cref{eq:energystability} for $\dot E$ gives
		\begin{equation}\label{eq:Edot_basic_bound}
			\dot{E}[\psi] = \sum_{j=0}^2 2\lambda_j|c_j|^2 + \sum_{n=0}^2\langle\cE\psi^b_n,\psi^b_n\rangle\leq \sum_{j=0}^2 2\lambda_j|c_j|^2 + \sum_{n=0}^22\kappa_n q_n^2,
		\end{equation}
where each $\kappa_{n}$ is the maximum eigenvalue of $\cE$ on the subspace containing  $\psi^b_n$ that is specified in~\cref{eq:tailcomponents}. Since the only energy-unstable modes in our regime of interest are $\varphi_{\pm 1}$, the only positive energy eigenvalue is $\lambda_1$, and so $\kappa_1$, $\kappa_2$, and $\kappa_3$ are strictly negative.
		
		The coefficients $c_j=-\langle\varphi_j,\nabla^2\psi\rangle$ can be differentiated using the vorticity evolution equation~\cref{eq:NS}, which gives
\begin{equation}
\label{eq:cdot_naive}
\begin{aligned}
\dot{c}_j &= \left\langle \varphi_j,\{\Psi + \psi,\zeta\} + \{\psi,Z\} - \nu\nabla^2\zeta\right\rangle \\
&=\langle \varphi_j,\psi,\zeta\rangle + \langle\varphi_j,\Psi,\zeta\rangle + \langle\varphi_j,\psi,Z\rangle - \nu\langle \varphi_j,\nabla^2\zeta\rangle,
\end{aligned}
\end{equation}
where we have introduced the triple product
		\begin{equation}\label{eq:tripleprod}
			\langle a,b,c\rangle= \langle a,\{b,c\}\rangle
		\end{equation}
for scalar fields. This triple product satisfies a number of identities that simplify our calculations. First, integration by parts gives
\begin{equation}
\langle\{a,b\},c\rangle = \langle a,\{b,c\}\rangle
\end{equation}
in several useful cases: (1) when at least two of the fields lie in $\cH_\CC$, which requires their gradients to vanish at the boundary, (2) when at least one of the fields vanishes at the boundary, or (3) when one field depends only on $y$ and another lies in $\cH_\CC$. In all of these cases, the triple product is anticommutative in all arguments,
\begin{equation}\label{eq:tripleprod_anticommute}
\langle a, b, c\rangle = -\langle b,a,c\rangle = -\langle a,c,b\rangle = -\langle c,b,a\rangle,
\end{equation}
and thus it is unchanged by cyclic permutations:
\begin{equation}\label{eq:tripleprod_shift}
			\langle a, b, c\rangle = \langle c, a, b\rangle = \langle b, c, a\rangle.
\end{equation}
For any fields $a,b,c$, integration by parts also gives the identity
\begin{equation}
\label{eq:tripleprod_integral}
			\langle a_x,b,c\rangle = -\langle a,b_x,c\rangle - \langle a,b,c_x\rangle,
\end{equation}
and likewise for $y$ derivatives, as well as
		\begin{equation}\label{eq:tripleprod_Qidentity}
			\langle a,b,\nabla^2c\rangle + \langle a,c,\nabla^2b\rangle = -\langle\nabla a,b,\nabla c\rangle - \langle \nabla a,c,\nabla b\rangle = \langle\nabla b \cdot Q_a\cdot\nabla c\rangle,
		\end{equation}
with $Q_a$ defined by \cref{eq:bigQ}. When any two of the fields are vectors, the triple product is defined by summing over the vector indices as in, e.g., $\langle \mathbf{a},\mathbf{b},c\rangle=\sum_j\langle a_j,b_j,c\rangle$. %The identities \cref{eq:tripleprod_anticommute}, \cref{eq:tripleprod_shift}, \cref{eq:tripleprod_integral}, and \cref{eq:tripleprod_Qidentity} hold element-wise for vector-valued arguments.
		
The expression~\cref{eq:cdot_naive} for $\dot c_j$ can be manipulated to another useful form by applying the identities \cref{eq:energystability,eq:tripleprod_Qidentity} and integrating by parts several times:
		\begin{equation}\label{eq:cdot_alt}
			\begin{aligned}
				\dot{c}_j &= \langle \varphi_j,\psi,\zeta\rangle + \langle\varphi_j,\Psi,\zeta\rangle + \langle\varphi_j,\psi,Z\rangle - \nu\langle \varphi_j,\nabla^2\zeta\rangle \\
				&= \langle \varphi_j,\psi,\zeta\rangle + \langle\varphi_j,\Psi,\zeta\rangle + \langle\varphi_j,\psi,Z\rangle + \langle\cE\varphi_j,\psi\rangle + \tfrac{1}{2}\langle\nabla \varphi_j\cdot Q_\Psi\cdot\nabla\psi\rangle \\
				&= \langle \varphi_j,\psi,\zeta\rangle + \langle\varphi_j,\Psi,\zeta\rangle + \langle\varphi_j,\psi,Z\rangle + \langle\cE\varphi_j,\psi\rangle - \tfrac{1}{2}\langle\nabla\Psi,\varphi_j,\nabla\psi\rangle - \tfrac{1}{2}\langle\nabla\Psi,\psi,\nabla\varphi_j\rangle\\
				&= \langle \varphi_j,\psi,\zeta\rangle - \langle\nabla\varphi_j,\Psi,\nabla\psi\rangle -\tfrac{1}{2} \langle\varphi_j,\nabla\Psi,\nabla\psi\rangle+ \langle\varphi_j,\psi,Z\rangle + \langle\cE\varphi_j,\psi\rangle - \tfrac{1}{2}\langle\nabla\Psi,\psi,\nabla\varphi_j\rangle\\
                &= \langle \varphi_j,\psi,\zeta\rangle + \tfrac{1}{2}\langle\varphi_j,\psi,Z\rangle  + \langle\cE\varphi_j,\psi\rangle + \langle\Psi,\nabla\varphi_j,\nabla\psi\rangle.
			\end{aligned}
		\end{equation}
	It is also helpful to decompose $\dot c_j$ by splitting $\psi=\psi^a+\psi^b$. Leveraging \cref{eq:tripleprod_Qidentity}, we find
		\begin{equation}\label{eq:cdot_d}
			\dot{c}_j = d_j[\psi^a] + d_j[\psi^b] + \langle \nabla\psi^a\cdot Q_{\varphi_j}\cdot\nabla\psi^b\rangle,
		\end{equation}
    where $d_j$ is defined as
    \begin{equation}
    \label{eq:dj}
    \begin{aligned}
				d_j[f] &= \langle \varphi_j,f,\nabla^2f\rangle + \langle\varphi_j,\Psi,\nabla^2f\rangle + \langle\varphi_j,f,Z\rangle - \nu\langle \varphi_j,\nabla^4f\rangle\\&= \langle \varphi_j,f,\nabla^2f\rangle - \lambda_j\langle\varphi_j,\nabla^2f\rangle + \langle \Psi, \nabla\varphi_j,\nabla f\rangle + \tfrac{1}{2}\langle Z,\varphi_j,f\rangle
			\end{aligned}
    \end{equation}
    for any $f\in\cH_\CC$. The first and second lines of~\cref{eq:dj} come from the $\dot c_j$ expressions in the last lines of~\cref{eq:cdot_naive,eq:cdot_alt}, respectively, as well as the identity $\cE\varphi_j = -\lambda_j\nabla^2\varphi_j$.

    With this notation in hand, as well as the decomposition $\psi=\psi^a+\psi^b$, we split the expression~\cref{eq:dotV} for $\dot{V}$ into `primary' and `tail' components,
        \begin{equation}
        \label{eq:Vab}
            \dot{V} = [\dot{V}]_a + [\dot{V}]_b,
        \end{equation}
        defined by
        \begin{equation}
        \label{eq:VaVb}
            \begin{aligned}
			[\dot{V}]_a &= \dot{E}[\psi^a]\left(2E + 2\gamma_0c_0 + \gamma_1\right) + 2\gamma_0E d_0[\psi^a] + 2\gamma_2\Re (A_{ij}\ovl{c_i}d_j[\psi^a]),\vspace{0.05in}\\
			[\dot{V}]_b &= \dot{E}[\psi^b]\left(2E + 2\gamma_0c_0 + \gamma_1\right) + 2\gamma_0E d_0[\psi^b] + 2\gamma_2\Re (A_{ij}\ovl{c_i}d_j[\psi^b])\\
			&\qquad + 2\gamma_0E \langle \nabla\psi^a\cdot Q_{\varphi_0}\cdot\nabla\psi^b\rangle + 2\gamma_2\Re (A_{ij}\ovl{c_i}\langle \nabla\psi^a\cdot Q_{\varphi_j}\cdot\nabla\psi^b\rangle).
            \end{aligned}
		\end{equation}
    In other words, $[\dot V]_a$ collects all terms in $\dot V$ that are fully determined by $\psi^a$, where $[\dot V]_b$ collects all terms that depend on $\psi^b$. We turn now to the problem of ensuring that $\dot{V}$ is negative definite.
        
    \subsection{Polynomial bounds on the time derivative}
    \label{sec:fourthterm}
    
    To verify that $\dot{V}[\psi]$ is negative definite, we bound it above by a polynomial in the six variables $(c_j,q_n)$. Then we can verify that this polynomial is negative definite by solving an SOS problem formulated in \cref{sec:finalsos}. The primary component $[\dot{V}]_a$ is already a polynomial in $(c_j,q_n)$, so bounds are needed only for the tail component $[\dot V]_b$.
    % We will discuss $[\dot{V}]_a$ in greater detail in \cref{sec:interpretation_primary}, to compare it against the argument of \cref{sec:interpretation}, but we turn now to the problem of bounding $[\dot{V}]_b$.
        
    There are five terms in the expression~\cref{eq:VaVb} for $[\dot{V}]_b$, all depending on the infinite-dimensional tail term $\psi^b$. We use similar analytical methods to derive upper bounds on each of the five terms. Here, we illustrate these methods by detailing the upper bound on the fourth term alone. Bounds for the other four terms are derived in \cref{app:estimates}.
    
    Starting with the integral in the fourth term of~\cref{eq:VaVb}, we expand $\psi^a$ and $\psi^b$ in their components, as in~\cref{eq:psia,eq:tailcomponents}, to find
        \begin{equation}
            \begin{aligned}
            \label{eq:4th term 1}
				\langle\nabla\psi^a\cdot Q_{\varphi_0}\cdot\nabla\psi^b\rangle &= \langle\nabla(2c_0\varphi_0 + c_i\ovl{\varphi_i} + \ovl{c_i}\varphi_i)\cdot Q_{\varphi_0}\cdot \nabla(\psi^b_0+\psi^b_1 + \psi^b_2)\rangle\\
                &= 2c_0\langle\nabla\varphi_0\cdot Q_{\varphi_0}\cdot \nabla\psi^b_0\rangle + 2\Re\langle c_i\nabla \ovl{\varphi_i}\cdot Q_{\varphi_0}\cdot \nabla\psi^b_{1,+}\rangle\\
                &=2\Re\langle c_i\nabla\ovl{\varphi_i}\cdot Q_{\varphi_0}\cdot \nabla\psi^b_{1,+}\rangle,
            \end{aligned}
		\end{equation}
        with summation over $i\in\{1,2\}$ implied. Most terms have vanished by orthogonality. In particular, $x$-integrals of $\nabla a_1\cdot Q_{a_2}\cdot \nabla a_3$ vanish for any triples of $a_i\in\cH_{k_i}$ with $k_1+k_2+k_3\neq 0$. Similarly, the $y$-integral of $\nabla\varphi_0\cdot Q_{\varphi_0}\cdot\nabla\psi^b_0$ must vanish. Noting also that $\|\nabla\psi^b_{1,+}\|=q_1$, we continue from~\cref{eq:4th term 1} and use Young's inequality~\cref{eq:young} to find
        \begin{equation}
        \label{eq:bound_to_compare}
            \begin{aligned}
                \left|\langle\nabla\psi^a\cdot Q_{\varphi_0}\cdot\nabla\psi^b\rangle \right|
                &\leq 2\| c_iQ_{\varphi_0}\cdot \nabla\ovl{\varphi_i}\| \,|q_1|\\
                &\leq \frac{1}{\eps_3}\|c_iQ_{\varphi_0}\cdot\nabla\ovl{\varphi_i} \|^2 + \eps_3q_1^2\\
                &=\frac{1}{\eps_3}\langle \ovl{c_i}\ovl{Q_{\varphi_0}}\cdot\nabla\varphi_i,c_jQ_{\varphi_0}\cdot\nabla\ovl{\varphi_j}\rangle + \eps_3q_1^2\\
                &=\frac{1}{\eps_3}\langle \ovl{Q_{\varphi_0}}\cdot\nabla\varphi_i,Q_{\varphi_0}\cdot\nabla\ovl{\varphi_j}\rangle \ovl{c_i}c_j + \eps_3q_1^2
            \end{aligned}
        \end{equation}
        for any $\eps_3>0$, where summation over $i,j\in\{1,2\}$ is implied. In all, the fourth term in the definition~\cref{eq:VaVb} of $[\dot V]_b$ is bounded according to
        \begin{equation}
        \label{eq:bound4}
            \left|2\gamma_0E\langle\nabla\psi^a\cdot Q_{\varphi_0}\cdot\nabla\psi^b\rangle\right| \leq 2|\gamma_0|E\left(\frac{1}{\eps_3}\langle \ovl{Q_{\varphi_0}}\cdot\nabla\varphi_i,Q_{\varphi_0}\cdot\nabla\ovl{\varphi_j}\rangle\ovl{c_i}c_j + \eps_3q_1^2\right).
        \end{equation}
        The left-hand side of~\cref{eq:bound4} depends on the full field $\psi=\psi^a+\psi^b$, whereas the right-hand side is determined by only the $c_j$ and $q_n$ variables (as well as choices of $\gamma_0$ and $\eps_3$). This has met our objective: to bound a term involving $\psi^b$ with one that is polynomial in a finite-dimensional space. Similar bounds for the other four terms in $[\dot{V}]_b$ appear in \cref{app:estimates} as \cref{eq:bound1}, \cref{eq:bound2}, \cref{eq:bound3}, and \cref{eq:bound5}.
        
        Combining these five bounds with the explicit expression for $[\dot V]_a$ yields an upper bound of the form
        \begin{equation}\label{eq:polybound_ours}
            \dot{V}[\psi]\leq S(c_j,q_n),
        \end{equation}
        where $S$ is a real-valued, quartic polynomial in the real variables $(c_0,q_0,q_1,q_2)$ and the complex variables $(c_1,c_2)$. The coefficients of this polynomial depend on the coefficients $(\gamma_0,\gamma_1,\gamma_2,\theta)$ in our ansatz~\cref{eq:V} for $V$, and on eight positive constants $\eps_0,...,\eps_7$ introduced via Young's inequality. More explicitly, the polynomial $S$ takes the form
        \begin{equation}
        \label{eq:fullpoly}\small
        	\begin{aligned}
        		S(c_j,q_n) &= 2\big(\lambda_m|c_m|^2+\kappa_n q_n^2\big)\left(2E + 2\gamma_0c_0 + \gamma_1\right)+ 2\gamma_0\lambda_0 Ec_0  + \left(\gamma_2B_{ij} + \gamma_0 C_{ij}E + \gamma_2D_{ij}c_0\right)\ovl{c_i}c_j\\
        		&\qquad +\gamma_2\omega q_1^2 + \gamma_0\tau_nEq_n^2+ \gamma_2\big(\eta_0 + \eta_j|c_j|^2\big)\big(\sqrt{2}q_0q_1 + q_1q_2 + q_2^2\big)+ \gamma_2\beta_{n,m}q_n|c_m|^2,
        	\end{aligned}
        \end{equation}
        where summation over $i,j\in\{1,2\}$ and $m,n\in\{0,1,2\}$ is implied. The coefficients $B_{ij}$, $C_{ij}$, $D_{ij}$, $\kappa_n$, $\eta_j$, $\beta_{n,m}$, and $\omega$ depend only on $\theta$ and $\eps_m$; expressions for all of them  are given in \cref{app:estimates}. The polynomial $S$ is always real, despite $c_1$ and $c_2$ being complex, because the matrices $B_{ij}$, $C_{ij}$, and $D_{ij}$ are Hermitian. 
        
        We will formulate an SOS condition in \cref{sec:finalsos} that implies negative definiteness of the polynomial $S(c_j,q_n)$ and, in turn, of $\dot V[\psi]$. In order to find values of the $\gamma_i$ and $\eps_m$ for which the SOS condition holds, it is ideal if they are optimization parameters of the SOS computation itself. Then, as the computation converges towards an SOS representation with each iteration, the $\gamma_i$ and $\eps_m$ also converge to values that make the SOS representation possible. However, in order for the SOS problem to be convex and solvable with standard methods, optimization parameters must only appear linearly in the SOS constraints. Our upper bounds on terms in $[\dot V]_b$ in \cref{sec:fourthterm,app:estimates}, which lead to $S$, fail to be linear in $\gamma_i$ and $\eps_m$, but this can be repaired in two ways. 
        
        The first impediment to linearity is that our bounds on $[\dot V]_b$ involve the absolute values $|\gamma_i|$. We fix this by anticipating that all optimal $\gamma_i$ values are positive. This allows us to add constraints that $\gamma_i\ge0$ and remove the absolute values in the $[\dot V]_b$ bounds, as has already been done in the expression~\cref{eq:fullpoly} for $S$. Note that signs of the optimal $\gamma_0$ or $\gamma_2$ can change if energy modes are defined with different sign conventions, while $\gamma_1\geq 0$ is required anyways for positive definiteness of $V$.

        The second impediment to linearity is the $\eps_m$ parameters, which have entered the bounds on $[\dot V]_b$ via Young's inequality. In these bounds, both $\eps_m$ and $1/\eps_m$ appear linearly, and in some cases they multiply $|\gamma_m|$ as well. This is reflected in the coefficients of $S$ given in \cref{sec:S coeff}. However, these particular nonlinearities from Young's inequality can be reformulated using semidefinite constraints that are linear in $\eps_m$ and $\gamma_i$. To see how, consider an abstract SOS setting with real, independent variables $\vec{z}\in\RR^p$. Young's inequality gives
        \begin{equation}
        \label{eq:young_SOS_example}
            \|\vec{v}(z_j)\|\,\|\vec{w}(z_j)\| \leq \frac{\eps}{2}\|\vec{v}(z_j)\|^2 + \frac{1}{2\eps}\|\vec{w}(z_j)\|^2
        \end{equation}
        for all $\eps>0$, where $\vec{v}(z_j)\in\RR^m$ and $\vec{w}(z_j)\in\RR^n$ are linear functions of $\vec{z}$ that involve no optimization parameters. Suppose that the right-hand terms are part of a polynomial expression that is constrained to be SOS, and that we want $\eps$ to be an optimization parameter in the SOS computation. This can be done be done by reformulating~\cref{eq:young_SOS_example} as
        \begin{equation}\label{eq:deltaeps_1}
            \|\vec{v}(z_j)\|\,\|\vec{w}(z_j)\| \leq \frac{\eps}{2}\|\vec{v}(z_j)\|^2 + \frac{\delta}{2}\|\vec{w}(z_j)\|^2,\qquad \begin{pmatrix}
                \eps & 1\\ 1& \delta
            \end{pmatrix}\succeq0.
        \end{equation}
        Then the right-hand side of~\cref{eq:deltaeps_1} can be used in place of the right-hand side of~\cref{eq:young_SOS_example}, provided the semidefiniteness constraint in~\cref{eq:deltaeps_1} is appended to the SOS problem. In cases where the term to be bounded is also multiplied by $|\gamma|$, where $\gamma$ is another optimization parameter, the bound~\cref{eq:deltaeps_1} can be generalized as
        \begin{equation}\label{eq:deltaeps_2}
            |\gamma|\, \|\vec{v}(z_j)\|\,\|\vec{w}(z_j)\| \leq \frac{\eps}{2}\|\vec{v}(z_j)\|^2 + \frac{\delta}{2}\|\vec{w}(z_j)\|^2,\qquad \begin{pmatrix}
                \eps & \gamma\\ \gamma& \delta
            \end{pmatrix}\succeq0 .
        \end{equation}
        For instance, this lets our bound~\cref{eq:bound4} on the fourth term of $[\dot V]_b$ be reformulated as
        \begin{equation}
            |\gamma_0|\left|\langle\nabla\psi^a\cdot Q_{\varphi_0}\cdot\nabla\psi^b\rangle\right|\leq \tilde{\delta}_3\langle \ovl{Q_{\varphi_0}}\cdot\nabla\varphi_i,Q_{\varphi_0}\cdot\nabla\ovl{\varphi_j}\rangle\ovl{c_i}c_j + \tilde{\eps}_3 q_1^2,\qquad \begin{pmatrix}
                \tilde{\eps}_3 & \gamma_0\\\gamma_0&\tilde{\delta}_3
            \end{pmatrix}\succeq 0,
        \end{equation}
        where we have changed variables by $\tilde{\eps}_3=\gamma_0\eps_3$. This approach is used in \cref{sec:finalsos} to reformulate all terms in $S$ that are not linear in the parameters $(\gamma_i,\eps_m)$.

        \subsection{Treating continua of Reynolds numbers}
        \label{sec:continua}

        A major limitation of previous works \citep{goulart_global_2012,fuentes_global_2022,iligaray_improved_2026} is that their SOS conditions can verify global stability at a specified Reynolds number $\nRe^*$, but they cannot imply the same for smaller Reynolds numbers $\nRe<\nRe^*$. The present subsection overcomes this limitation as follows. Suppose a chosen $V$ is a Lyapunov functional at a fixed viscosity value $\nu^*$, meaning that it satisfies
        \begin{equation}
        \label{eq:V before dV/dnu}
            \dot{V}[\psi] < 0\qquad \text{for}\;\psi\neq 0
        \end{equation}
        when $\nu=\nu^*$. Recall that $\dot{V}[\psi]$ depends on $\nu$ as well as the streamfunction. If one can further show that
        \begin{equation}
        \label{eq:dV/dnu condition}
            \frac{\d}{\d\nu} \dot{V}[\psi] \leq 0
        \end{equation}
        for all $\nu\geq\nu^*$, this implies~\cref{eq:V before dV/dnu} at such $\nu$. It would follow that $V$ is a Lyapunov function at all $\nu\ge\nu^*$, i.e., at all $\nRe\le\nRe^*$.

        The additional condition~\cref{eq:dV/dnu condition} is relatively simple to verify---for the present fluid models and many others---because the viscosity necessarily appears linearly in $\dot V$. For any fixed functional $V$ on $\cH_\RR$, its time derivative takes the form
        \begin{equation}
        \label{eq:F1 F2}
            \dot{V}[\psi] = F_1[\psi] + \nu F_2[\psi]
        \end{equation}
        for $\nu$-independent functionals $F_1$ and $F_2$, and so $\frac{\d}{\d\nu}\dot{V}[\psi]=F_2[\psi]$. 
        %If one can show that $F_2$ is everywhere non-positive, then for any values $\nu>\nu^*>0$, one has
        %\begin{equation}
        %    \dot{V}[\psi] = F_1[\psi] + \nu F_2[\psi] = (F_1[\psi] + \nu^*F_2[\psi]) + (\nu-\nu^*)F_2[\psi] \leq F_1[\psi] + \nu^*F_2[\psi]
        %\end{equation}
        %for all $\psi\in\cH_\RR$. In particular, if $\dot{V}$ is negative definite at the viscosity $\nu^*$ and $F_2$ is non-positive, it follows that $\dot{V}$ is negative definite for all $\nu>\nu^*$. 
        We emphasize that, in our context, $V$ is constructed using $\varphi_j$ that are energy modes at viscosity $\nu^*$, so the only quantity in~\cref{eq:F1 F2} that changes with viscosity is $\nu$ itself.

        For our chosen $V$ ansatz, the expression for $F_2[\psi]$ can be found from \cref{eq:dotV,eq:cdot_alt,eq:edot} to be
        \begin{equation}
            \frac{\d}{\d\nu}\dot{V}[\psi] = -\|\zeta\|^2(2E + 2\gamma_0c_0 + \gamma_1) - 2\gamma_0 E\langle\varphi_0,\nabla^2\zeta\rangle - 2\gamma_2\Re\left[A_{ij}\ovl{c}_i\langle\varphi_j,\nabla^2\zeta\rangle\right].
        \end{equation}
        We prove that this quantity is negative definite in a somewhat different way from how we handled $\dot{V}$ in \cref{sec:fourthterm}. First, we define $Z_j = \|\nabla^2\varphi_j\|^2$ for $j\in\{0,1,2\}$, so that
        \begin{equation}
            |\langle\varphi_j,\nabla^2\zeta\rangle| = |\langle\nabla^2\varphi_j,\zeta\rangle| \leq \sqrt{Z_j}\,\|\zeta\|.
        \end{equation}
        Applying Young's inequality several times gives
        \begin{equation}
            \begin{aligned}
                \frac{\d}{\d\nu}\dot{V}[\psi] &\leq -2\|\zeta\|^2 E - \gamma_1\|\zeta\|^2 + |\gamma_0|\|\zeta\|^2(\eps_8|c_0|^2 +1/\eps_8)\\
                &\qquad + |\gamma_0|E(\eps_9Z_0\|\zeta\|^2 + 1/\eps_9) + 2|\gamma_2|\left(|c_1|\sqrt{Z_2} + |c_2|\sqrt{Z_1}\right)\|\zeta\|\\
                &\leq -2\|\zeta\|^2 E - \gamma_1\|\zeta\|^2 + |\gamma_0|\|\zeta\|^2(\eps_8|c_0|^2 +1/\eps_8)\\
                &\qquad + |\gamma_0| E(\eps_9Z_0\|\zeta\|^2 + 1/\eps_9) + |\gamma_2|\eps_{10}\big(|c_1|^2+ |c_2|^2\big) \\
                &\qquad + \tfrac1{\eps_{10}}|\gamma_2|\left(Z_2 + Z_1\right)\|\zeta\|^2
            \end{aligned}
        \end{equation}
        for any $\eps_m>0$, where these three $\eps_m$ are to be chosen along with the seven in the coefficients of $S$. Next, Poincar\'e's inequality yields
        \begin{equation}
            E\leq \frac{1}{2\pi^2} \|\zeta\|^2,\qquad |c_1|^2 + |c_2|^2 \leq \frac{1}{2\left[\pi^2 + (2\pi/L)^2\right]}\|\zeta\|^2,
        \end{equation}
        in which the latter has an improved constant because streamwise-constant modes are excluded. Finally, recalling that all $\gamma_i\ge0$, as explained at the end of \cref{sec:fourthterm}, we recover the inequality
        \begin{equation}\label{eq:Vdot_nu}
            \frac{\d}{\d\nu}\dot{V}[\psi] \leq -s_1 \|\zeta\|^2E - s_2\|\zeta\|^2,
        \end{equation}
        where
        \begin{equation}
        \begin{aligned}
        \label{eq:si def}
            s_1 &= 2 - \gamma_0\eps_8 - \gamma_0Z_0\eps_9,\\
            s_2 &= \gamma_1 - \gamma_0/\eps_8 - \frac{1}{2\pi^2}\gamma_0/\eps_9 - \frac{\gamma_2\eps_{10}}{2\left[\pi^2 + (2\pi/L)^2\right]} - \gamma_2(Z_2 + Z_1)/\eps_{10}.
        \end{aligned}            
        \end{equation}
    If one can verify that $s_1,s_2\geq0$ and that $\dot{V}$ is negative definite at a fixed $\nRe^*$, then \cref{eq:Vdot_nu} implies that $\dot{V}$ is negative definite at all $\nRe\leq\nRe^*$, and global stability follows at all of these~$\nRe$. We note that the optimal value of $\eps_{10}$ could be deduced \emph{a priori} from the values of $Z_1$ and $Z_2$, if desired.
    
    Following the strategy described at the end of \cref{sec:fourthterm}, one can ensure that $s_1$ and $s_2$ depend only linearly on a larger set of optimization parameters $(\gamma_m,\eps_m,\delta_m)$, which lets the inequalities $s_1,s_2\geq 0$ be incorporated into a larger SOS problem in the next subsection. We do not need the notation $\nRe^*$ below, so we revert to $\nRe$.
         
     %We further manipulate the expressions~\cref{eq:si def} such that optimization parameters appear only linearly, because this is needed in the next subsection to enforce nonnegativity of $s_1$ and $s_2$ in an SOS computation. For $\eps_{10}$, one can \emph{a priori} choose the optimal value that maximizes the last two terms in $s_2$ at fixed $\gamma_2$. For $\eps_8$ and $\eps_9$, we follow the approach described at the end of \cref{sec:fourthterm}, letting $\tilde\eps_m=\gamma_0\eps_m$ and $\tilde\delta_m\ge\gamma_0/\eps_m$. Then, 
     %\begin{equation}
     %\begin{aligned}
     %\label{eq:si def}
     %s_1 \ge \tilde s_1 =& 2 - \tilde\eps_8 - Z_0\tilde\eps_9,\\
     %s_2 \ge \tilde s_2 =& \gamma_1 - \tilde\delta_8 - \frac{1}{2\pi^2}\tilde\delta_9 - \gamma_2\sqrt{\frac{2(Z_1+Z_2)}{\pi^2 + (2\pi/L)^2}}
    %\end{aligned}          
    %\end{equation}
    %for any $\tilde\eps_m$ and $\tilde\delta_m$ that satisfy $M_m\succeq0$ with
    %\begin{equation}
    %\label{eq:Mm}
    %M_m=\begin{bmatrix}
    %\tilde\eps_m&\gamma_0\\ \gamma_0&\tilde\delta_m
    %\end{bmatrix}.
    %\end{equation}
    %Rather than enforcing $s_1,s_2\ge0$, we will enforce $\tilde s_1,\tilde s_2\ge0$ along with $M_m\succeq0$.
        
		\subsection{Sum-of-squares optimization}
        \label{sec:finalsos}
        
		Although the expressions involved in the key bounds \cref{eq:polybound_ours} and \cref{eq:Vdot_nu} are complicated, an SOS procedure can ensure that 
		\begin{equation}
        \label{eq:constraints before sos}
		S(c_j,q_n)\leq -rE[\psi]^2\qquad\text{and}\qquad s_1,s_2\geq 0
		\end{equation}
        for all $\psi\in\cH_\RR$ at some tolerance $r>0$ and some chosen pair of $(L,\nRe)$ values. As discussed in the preceding subsections, verifying~\cref{eq:constraints before sos} implies global stability for our system at the streamwise period $L$ and all Reynolds numbers up to $\nRe$.

The SOS procedure is carried out as follows. We strengthen the first inequality in \cref{eq:constraints before sos} into an SOS condition, meaning we require that $-S - rE^2$ can be written as a sum of squares of lower-degree polynomials:
        \begin{equation}
            -S(c_j,q_n)^2 - rE[\psi]^2 = \sum_{m=1}^M \left[p_m(c_j,q_n)\right]^2.
        \end{equation}
This condition is stronger than negative definiteness of $S$, but also much more tractable to verify; SOS conditions can be recast as semidefinite programs \citep{parrilo_semidefinite_2003}, which can be solved numerically using interior point algorithms. Notably, such algorithms can optimize values of parameters that appear linearly in $S$, such as the coefficients $\gamma_j$ of our candidate functional \cref{eq:V} and the parameters $\eps_m$ and $\delta_m$ arising from our use of Young's inequality, as discussed at the end of \cref{sec:fourthterm}. Finally, SOS optimization also allows for semidefinite constraints that are linear in the optimization parameters, so we can include the $\gamma_j\geq 0$ constraints discussed in \cref{sec:fourthterm}, as well as the constraints $M_m\succeq 0$ arising from each use of Young's inequality, where $M_m(\eps_m,\delta_m)$ are $2\times 2$ matrices of the form \cref{eq:deltaeps_1} or \cref{eq:deltaeps_2}. The only parameter that cannot be optimized during the SOS computation is the phase $\theta$ appearing in our ansatz \cref{eq:V}, which we discuss shortly. In all, our SOS problem reads as follows:
        \begin{equation}
        \label{eq:rsos}
            r_\text{SOS}(\nRe, L,\theta) = \max_{\substack{(\gamma_j,\eps_m,\delta_m)}} r\in\RR\quad\text{s.t.}\quad\begin{array}[t]{l}
        -S-rE^2\;\text{is SOS} \\
        s_1,s_2\geq 0\\
        \gamma_j\geq0\\
        M_m\succeq0.
        \end{array}
        \end{equation}
        The maximum value $r_\text{SOS}(L,\nRe,\theta)$ is inexpensive to compute using standard software, taking approximately $0.3$ seconds on a laptop. If the value of $r_\mathrm{SOS}$ is strictly positive for some value of $\theta$, it follows that flows with streamwise period $L$ must be globally stable up to $\nRe$. For Couette flow, $\theta=0$ can be fixed \emph{a priori}, following \cref{lem:symmetry_specific}. For Poiseuille flow, we carry out a nonlinear optimization over $0<\theta\leq 2\pi$ by solving the SOS problem~\cref{eq:rsos} repeatedly with different $\theta$ values. However, we find empirically that the resulting value of $\theta$ is always such that the following integral is real and positive:
        \begin{equation}
            \mu = \langle\Psi,\nabla\varphi_2,\nabla\ovl{\varphi_1}\rangle + \tfrac{1}{2}\langle Z, \varphi_2,\ovl{\varphi_1}\rangle.
        \end{equation}
This choice of $\theta$ ensures that $[\dot{V}]_2$ is maximally negative definite when restricted to the $\cH^*$ subspace, in the sense that $[\dot V]_2\le -cE$ for all $\psi\in\cH^*$ with the largest possible $c$. Note that one could fix the value of $\theta$ in 2-D Poiseuille flow \emph{a priori} by requiring $\mu\geq 0$, and thus sidestep the nonlinear optimization over $\theta$.

        \begin{figure}
            \centering
            \includegraphics[width=0.85\linewidth]{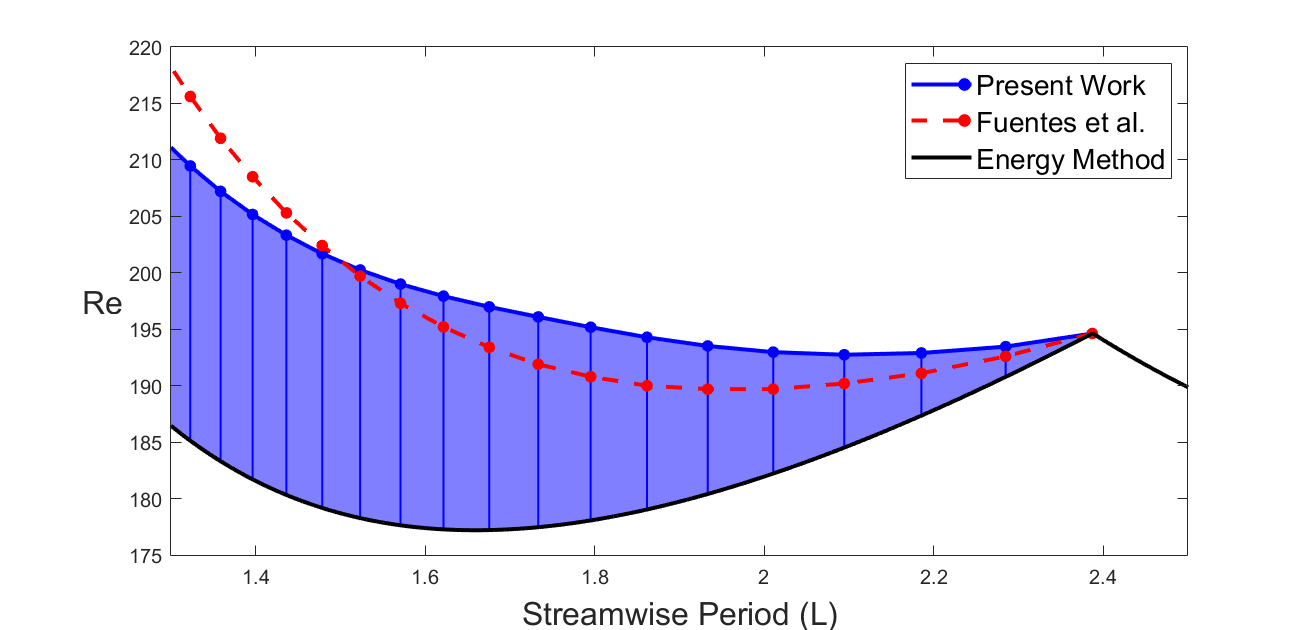}
            \caption{Global stability results for 2-D plane Couette flow at different values of the streamwise period $L$. The energy method gives global stability below the black curve. Our present approach gives global stability below the blue curve at each $L$ value where we have carried out computations, which are marked with vertical blue lines. Regions between these $L$ values are shaded for visual clarity. The method of \citet{fuentes_global_2022} gives global stability along the red curve at discrete $(L,\nRe)$ values where they carried out computations---though not below this curve---at much larger cost than our analysis (see text).}
            \label{fig:placeholder}
        \end{figure}

We present concrete stability results for both Couette and Poiseuille flow below. In each case, we use \texttt{YALMIP} \citep{lofberg_yalmip_2004,lofberg_pre_2009} to parse the problem \cref{eq:rsos} and turn it into a semidefinite program, which we then solve with \texttt{MOSEK} \citep{mosek}. We use a standard bisection search to find the maximum $\nRe$ such that $r_\mathrm{SOS}\geq 10^{-3}$. \Cref{fig:placeholder} summarizes our results for 2-D Couette flow. Our stability curve (blue) clearly improves upon the energy method curve (black), and over a range of $L$ it is also better than the curve computed by \citeauthor{fuentes_global_2022} using six energy modes (red). Our curve and that of \citeauthor{fuentes_global_2022} both coincide with the energy stability curve when $L\geq 2.387$, where a different energy-unstable mode arises. Progress in that regime requires a small modification of $V$ that we leave for future work. In order to give a concrete example of a quartic Lyapunov function, let us consider the streamwise period $L=1.659$ at which the energy Reynolds number $\nRe_E=177.2$ is minimized \citep{orr_stability_1907}. By solving the SOS problem~\cref{eq:rsos}, we deduce that $V$ of our quartic form~\cref{eq:lyapunov_primitive} is a Lyapunov functional for all $\nRe \leq 197.3$ with
\begin{equation}\label{eq:goodparams_C}
    \gamma_0 = 0.0501,\qquad \gamma_1 = 0.00680,\qquad \gamma_2 =  0.000484.
\end{equation}

For 2-D Poiseuille flow, we have used our computational procedure to verify global stability at various $L$, but the improvements over the energy method are more modest than for Couette flow. The $\nRe$ values up to which we can verify global stability of Poiseuille flow are only about 1--2\% larger than the energy stability values. Preliminary investigations with the code from \citet{iligaray_improved_2026} suggest that our relatively weak results for Poiseuille flow can be improved with the same ansatz by more carefully bounding certain modes of the tail. We leave this direction of investigation to future work. In order to give a concrete Lyapunov function for Poiseuille flow, let us consider the streamwise period $L=1.497$ at which the energy Reynolds number $\nRe_E=175.1$ is minimized. By solving the SOS problem~\cref{eq:rsos}, we deduce that \cref{eq:lyapunov_primitive} is a Lyapunov functional for $\nRe \leq 177.0$ with
\begin{equation}\label{eq:goodparams_P}
\gamma_0 = 0.0266,\qquad \gamma_1 = 0.00690,\qquad \gamma_2 =  0.00006197.
\end{equation}
		
        %Purely numerical proofs (Iligaray \& Fuentes, private communication) using our same ansatzes show that $\nRe_G>211$ for 2-D Couette flow at $L=1.659$, and $\nRe_G> 185$ for 2-D Poiseuille flow at $L=1.497$. In both cases, numerical experiments indicate that the bound cannot be further improved by lifting the ansatz \cref{eq:lyapunov_primitive}, at least without introducing additional energy modes.

        \subsection{Comparing against the strategy of Fuentes \emph{et al.}}\label{sec:comparison}
        
       Our approach to bounding $\dot V$ differs in one significant way from analogous bounds in \citet{fuentes_global_2022} and other past works. For an example, consider the term we bound using Young's inequality in~\cref{eq:bound_to_compare}. The strategy of \citeauthor{fuentes_global_2022}, roughly translated to the present notation, would be to bound the top line of~\cref{eq:bound_to_compare} by introducing an auxiliary polynomial $p_1(c_j,q_n)$ such that
        \begin{equation}\label{eq:aux_poly}
        2\|c_i\nabla\ovl{\varphi_i}\cdot Q_{\varphi_0}\| \,|q_1|  \leq p_1(c_j,q_n).
        \end{equation}
One can avoid square roots (which are incompatible with SOS computation) by squaring both sides:
        \begin{equation}\label{eq:constraint_schur1}
        [p_1(c_j,q_n)]^2 - (\vec{c}^\dagger\vec{M}\vec{c}) q_1^2\geq 0, \qquad [\vec{M}]_{ij} = 4\langle \ovl{Q_{\varphi_0}}\cdot\nabla\varphi_i,Q_{\varphi_0}\cdot\nabla\ovl{\varphi_j}\rangle,
        \end{equation}
        writing $\vec{c} = (c_1,c_2)^t$. To restore linearity in $p_1$, one observes that the left-hand side of \cref{eq:constraint_schur1} is proportional to the Schur complement of the matrix
        \begin{equation}
        \label{eq:L1}
        \vec{L}_1(c_j,q_n)=\begin{pmatrix}
            (\vec{c}^\dagger\vec{M}\vec{c})p_1(c_j,q_n) & (\vec{c}^\dagger\vec{M}\vec{c})q_1\\
            (\vec{c}^\dagger\vec{M}\vec{c})q_1& p_1(c_j,q_n)
        \end{pmatrix},
        \end{equation}
        so the condition \cref{eq:constraint_schur1} is equivalent to ensuring that $\vec{L}_1(c_j,q_n)\succeq 0$ for all $(c_j,q_n)$. This constraint is linear in the unknown coefficients of $p_1(c_j,q_n)$, so it can be strengthened into a matrix-valued SOS condition. The left-hand side of~\cref{eq:bound_to_compare} is therefore bounded above by any $p_1(c_j,q_n)$ that satisfies~\cref{eq:L1}, so one can replace this term by $p_1$ when deriving an upper bound on $\dot{V}$.
        % Meanwhile, the  $p_1(c_j,q_n)$ would take the place of the expression \cref{eq:fourthterm} in a bound of the form \cref{eq:intro_polybound}, or heuristically,
        % \begin{equation}
        %     S(c_j,q_n)\sim \dot{V}[\psi] - 2\gamma_0E[\psi] \langle \nabla\psi^a\cdot Q_{\varphi_0}\cdot\nabla\psi^b\rangle + 2\gamma_0E[\psi] p_1(c_j,q_n),
        % \end{equation}
        % which is now `one step closer' to being a polynomial in $(c_j,q_n)$. 
        Such a procedure can be repeated for every non-polynomial term in $\dot V$ that must be bounded, each time introducing another auxiliary polynomial $p_i$ and corresponding matrix SOS constraint $\vec{L}_i(c_j,q_n)\succeq 0$. This lets $\dot V$ be bounded above by an expression that is polynomial in $(c_j,q_n)$, but which contains many polynomials $p_i$ and must be accompanied by expensive matrix-valued SOS constraints.
        
In more general contexts, it is possible that the approach of \citet{fuentes_global_2022} could give sharper estimates than our use of Young's inequality, albeit at much higher computational cost. For the estimates in our present work, however, it turns out that the two approaches are equally sharp. This claim is justified by the following proposition, which we prove in \cref{app:youngs}.
        \begin{proposition}
        \label{prop:young}
            Suppose $\vec{A}\in\RR^{m_1\times n}$ and $\vec{B}\in\RR^{m_2\times n}$, and consider a non-negative, homogeneous, quadratic polynomial $p(\vec{z})$ on $\RR^n$ satisfying
            \begin{equation}\label{eq:generic_poly}
                |\vec{A}\vec{z}|\,|\vec{B}\vec{z}|\leq  p(\vec{z}),\qquad \forall \vec{z}\in\RR^n.
            \end{equation}
            If $\vec{A},\vec{B}\neq 0$, then there exists an $\eps>0$ such that
            \begin{equation}
            \label{eq:youngs prop}
                |\vec{A}\vec{z}|\,|\vec{B}\vec{z}|\leq \frac{\eps}{2}|\vec{A}\vec{z}|^2 + \frac{1}{2\eps}|\vec{B}\vec{z}|^2\leq p(\vec{z}) .
            \end{equation}
            The same statement holds if $\vec{A}\neq 0$ and $\vec{B}\vec{z}$ is everywhere replaced with $1$.
        \end{proposition}

 %All of our applications of Young's inequality are in a form to which \cref{prop:young} applies, while about half the bounds by \citet{fuentes_global_2022} are not in such a form.
 
 Bounds of this form encompass all of our own applications of Young's inequality, as well as about half of the bounds carried out by \citet{fuentes_global_2022}. Their other bounds pose more serious difficulties for the Young's inequality approach; for instance, they construct several bounds of the form
        \begin{equation}\label{eq:bound_superbad}
            |\vec{A}\vec{z}|\,|\gamma P(\vec{z})+ \gamma'P'(\vec{z})|\le p(\vec{z}),
        \end{equation}
        where $P$ and $P'$ are fixed polynomials, $\gamma,\gamma'\in\RR$ are optimization parameters, and $p(\vec{z})$ is an auxiliary polynomial. Attempting to replace $p(\vec{z})$ with an application of Young's inequality would create quadratic dependence on $\gamma$ and $\gamma'$ that is incompatible with SOS computation. We can avoid situations like \cref{eq:bound_superbad} in the present work, because our restricted $V$ ansatz~\cref{eq:V} ensures that our equivalents of $P$ and $P'$ are linear, and we use triangle inequalities to separate (our equivalents of) $\gamma$ and $\gamma'$ and handle them individually. In any case, we believe the Young's inequality approach is more promising for both analytical and numerical investigations of global stability going forward.
         
        % We are optimistic that the techniques introduced in the present work might generalize to handle such cases, but further discussion is outside the present scope.
		
		\section{Discussion and perspectives}\label{sec:conclusion}
		We have here proposed an interpretable, quasi-analytical approach to prove global stability beyond the energy Reynolds number. To this end, we have introduced a general schematic \cref{eq:lyapunov_proposed} for candidate Lyapunov functions for incompressible flows, and we have discussed how such functionals may leverage physical properties of the system to meet the requirements of a Lyapunov function. Connecting this schematic with inspiration from the numerical results of \citet{fuentes_global_2022} and a further symmetry argument, we arrived at the simplest-possible class of candidate Lyapunov functions for 2-D shear flows: the three-parameter family \cref{eq:lyapunov_primitive}. We have used these functionals to verify global stability of 2-D plane Couette flow and Poiseuille flow at various streamwise periods $L$, for all $\nRe$ up to certain values $\nRe^*(L)>\nRe_E$. Previous efforts to surpass the energy method could verify stability only for a single $\nRe^*$ value at once, rather than for all $\nRe$ up to $\nRe^*$. In this sense, our results constitute the first lower bounds on the global stability Reynolds number $\nRe_G$---for any fluid flow---that surpass the energy method value.
		
		The present work offers a step towards a flexible and interpretable approach to improving global stability bounds in more general systems. Moreover, we believe that several ideas introduced here can be incorporated in future numerical efforts for significant cost savings. For one, a costly (and complicated) step in the algorithm of \citet{fuentes_global_2022} was reducing the space of admissible polynomials to account for system symmetries, primarily via automatic symmetry detection using the software \texttt{YALMIP} \citep{lofberg_yalmip_2004,lofberg_pre_2009}. In their formulation with real variables, \texttt{YALMIP} was only able to fully exploit streamwise translations by $L/2$, leaving the space of polynomials larger than necessary. By contrast, arbitrary streamwise translations and up-down symmetries (for both Couette and Poiseuille flow) are folded into our revised notation, yielding a far smaller, complex-valued SOS problem and eliminating most pre-processing. Similarly, another costly element in \citet{fuentes_global_2022} was the inclusion of several matrix-valued SOS constraints, similar to \cref{eq:L1}. We have shown that a large class of such bounds can be reduced to an appropriate version of Young's inequality (\cref{prop:young}), which can be incorporated in a larger SOS problem with negligible computational overhead.
        
        Finally, numerical experiments for both Couette and Poiseuille flow indicate that, even with many energy modes, Lyapunov functions within the schematic \cref{eq:lyapunov_proposed} yield the same bounds on $\nRe_G$ as fully general quartic Lyapunov functions depending on the same modes. This suggests restricting to the class \cref{eq:lyapunov_proposed}, which leads to far fewer possible monomials in the SOS constraints. Exploiting such sparsity of monomial terms can significantly reduce computational cost of SOS computations \citep{zheng_chordal_2021}, which we believe is a promising future direction.

		\subsection*{Acknowledgments}
		We thank Federico Fuentes, Giovanni Fantuzzi, Vicente Iligaray, and  Karol\'ina Sehnalov\'a for helpful discussion of this work, and Vicente and Karol\'ina in particular for careful comments on our manuscript and code. DD would like to extend his gratitude to the GFD Fellowship Program at Woods Hole Oceanographic Institute, which made the present research possible and helped shape his research more broadly. EC was supported in part by the Department of Defense Vannevar Bush Faculty Fellowship, under ONR award N00014-22-1-2790. DG was supported by the NSERC Discovery Grants Program (awards RGPIN-2018-04263 and RGPIN-2025-06823). Collaboration between DD and DG was facilitated by travel funding from the Pacific Institute for the Mathematical Sciences.
        
\renewcommand{\theHsection}{A\arabic{section}}
		\appendix

        \section{Construction of the function space $\cH_\RR$}\label{app:space}
        The function space $\cH_\RR$ defined by \cref{eq:cH} is non-standard, so we here provide the functional-analytical details of its construction. We begin with the restricted Sobolev space $H_0^1(\Omega)$, which contains functions $\phi$ that satisfy $\phi,\nabla\phi\in L^2$ and zero Dirichlet conditions on $\partial\Omega$. If $\partial\Omega$ is smooth, it can be shown that there exists a continuous map \citep[p.~50--51]{Sohr2012}
		\begin{equation}\tau_N:H_0^1(\Omega)\to H^{-1/2}(\partial\Omega),\qquad f\mapsto \hat{n}\cdot\nabla f,
        \end{equation}
		which extends the Neumann boundary trace of smooth functions. Here, $H^{-1/2}$ denotes a Sobolev space of negative order, whose elements are generalized functions. Similar statements hold for domains with corners, but we do not require them here. In particular, one can consider the kernel $\ker(\tau_N)\subset H_0^1(\Omega)$ of functions with zero Neumann conditions on $\partial\Omega$. This kernel enforces both Dirichlet and Neumann conditions on $\partial\Omega$, but we want to allow our streamfunctions to take constant, nonzero values on the two walls; physically, this corresponds to nonzero net flow through the channel. To correct this restriction, we augment our space with the single function $f(x,y) = \sin(\pi y)$. In set summation notation, we thus find
		\begin{equation}
		\cH_\RR = \ker(\tau_N)+ \op{span}\left[(x,y)\mapsto \sin(\pi y)\right].
        \end{equation}
        The space $\cH_\CC$ is simply the complexification of $\cH_\RR$.
        
        \section{Estimates on the tail of $\dot{V}$}
        \label{app:estimates}
        
        We here bound each term in $\dot{V}[\psi]$ by a polynomial expression in $(c_j,q_n)$. These estimates lead to an overall polynomial bound of the form $\dot{V}[\psi]\leq S(c_j,q_n)$, where the final expression for $S$ is reported in~\cref{eq:fullpoly}. Recalling the decomposition $\dot{V}=[\dot{V}]_a + [\dot{V}]_b$ given in \cref{eq:VaVb}, we note that the terms grouped into $[\dot V]_a$ are already polynomial in $(c_j,q_n)$, so they do not need to be estimated. There are five terms in the expression \cref{eq:VaVb} for $[\dot{V}]_b$; the fourth is treated in \cref{sec:fourthterm}, yielding the bound \cref{eq:bound4}, and we handle the remaining four in order below.
        
        \subsection[First term of Vb]{First term of $[\dot{V}]_b$}
        
        The first term in $[\dot{V}]_b$ depends on the tail $\psi^b$ only through the expression $\dot{E}[\psi^b$], so we only need to bound $\dot{E}[\psi^b]$ above. For this we recall \cref{eq:Edot_basic_bound}, noting that all $c_j$ are zero on the subspace containing $\psi^b$. This yields a polynomial upper bound on the first term of $[\dot{V}]_b$:
        \begin{equation}
        \label{eq:bound1}
            \dot{E}[\psi^b](2E + 2\gamma_0c_0 +\gamma_1) \leq 2\left(\sum_{n=0}^2\kappa_nq_n^2\right)(2E + 2\gamma_0c_0 +\gamma_1).
        \end{equation}

        \subsection[Second term of Vb]{Second term of $[\dot{V}]_b$}
        
        For the second term, we start by simplifying the expression~\cref{eq:dj} for $d_j[\psi^b]$, using the orthogonality of $\psi^a$ and $\psi^b$:
        \begin{equation}
        \label{eq:djpsib}
            d_j[\psi^b] = \langle\varphi_j,\psi^b,\nabla^2\psi^b\rangle + \langle\Psi,\nabla\varphi_j,\nabla\psi^b\rangle + \tfrac{1}{2}\langle Z,\varphi_j,\psi^b\rangle,
        \end{equation}
    In the $j=0$ case, we further apply \cref{eq:tripleprod_Qidentity} and retain only terms with compatible wavenumbers, finding
    \begin{equation}
    \label{eq:tail1}
            d_0[\psi^b] = \langle\varphi_0,\psi^b,\nabla^2\psi^b\rangle = \langle\nabla\psi^b_{1,-}\cdot Q_{\varphi_0}\cdot \nabla\psi^b_{1,+}\rangle + \langle\nabla\psi^b_{2,-}\cdot Q_{\varphi_0}\cdot \nabla\psi^b_{2,+}\rangle.
    \end{equation}
    Note that \cref{eq:tripleprod_Qidentity} gives rise to a natural `$2$-$2$-$\infty$' inequality for any fields $a,b,c$:
		\begin{equation}
        \label{eq:opnorm}
			\left|\langle a,b,\nabla^2c\rangle + \langle a,c,\nabla^2b\rangle \right| \leq \|Q_a\|_\mathrm{op}\|\nabla b\|\|\nabla c\|,
		\end{equation}
        where
        \begin{equation}
			\|Q_a\|_\mathrm{op}=\max_{(x,y)}\|Q_a(x,y)\|_{\mathrm{op},2\times 2}
		\end{equation}
		denotes the spatial maximum of the pointwise operator norm (in the sense of Euclidean norms) of the $2\times2$ matrix $Q_a(x,y)$. Since $Q_a(x,y)$ is self-adjoint, this pointwise operator norm is precisely the largest absolute value among its pointwise eigenvalues. 
        
        We can write $\varphi_0$ and $Q_{\varphi_0}$ explicitly for 2-D planar shear flows, because streamwise-constant energy modes are simply sinusoids in the wall-normal coordinate $y$. In Couette flow,
		\begin{equation}
        \varphi_0(y) = \frac{\cos(2\pi y)}{2\pi\sqrt{L}},\qquad Q_{\varphi_0} = -\frac{2\pi}{\sqrt{L}}\begin{pmatrix}
				& \cos(2\pi y)\\
				\cos(2\pi y)&
			\end{pmatrix},
		\end{equation}
        so $\|Q_{\varphi_0}\|_\mathrm{op} =2\pi L^{-1/2}$. The Poiseuille flow case is similar, and yields $\|Q_{\varphi_0}\|_\mathrm{op} =\pi L^{-1/2}$. Using the larger of the two values (for the sake of consistency), we bound the right-hand side of~\cref{eq:tail1} as in~\cref{eq:opnorm} and arrive at the following bound on the second term in $[\dot{V}]_b$:
        \begin{equation}\label{eq:bound2}
                \left|2\gamma_0 Ed_0[\psi^b]\right| \leq 4\pi L^{-1/2}|\gamma_0| E(q_1^2+q_2^2).
        \end{equation}
        We note that the bound \cref{eq:bound2} could be sharpened by solving Euler--Lagrange equations to compute the maximum of the right-hand side of~\cref{eq:tail1} over the $\psi$ components, rather than using the estimate~\cref{eq:opnorm}. We have carried out numerical experiments for Couette flow using this sharper bound, and we find that it increases the $\nRe$ values at which we verify stability, but by 1.0 or less. To keep our estimates simpler and applicable to other flows, we use~\cref{eq:bound2} for our main results.

        \subsection[Third term of Vb]{Third term of $[\dot{V}]_b$}
        
        The third term of $[\dot{V}_b]$ can be written more explicitly using \cref{eq:djpsib}:
        \begin{equation}
        \label{eq:thirdterm1}
        2\gamma_2\Re (A_{ij}\ovl{c_i}d_j[\psi^b]) = 2\gamma_2\Re [A_{ij}\ovl{c_i}(\langle\varphi_j,\psi^b,\nabla^2\psi^b\rangle+\langle \Psi,\nabla\varphi_j,\nabla\psi^b\rangle + \tfrac{1}{2}\langle Z,\varphi_j,\psi^b\rangle)].
        \end{equation}
        In the first integral on the right-hand side, we expand $\psi$ into components and use orthogonality to retain only compatible wavenumbers:
		\begin{equation}
			\langle\varphi_j,\psi^b,\nabla^2\psi^b\rangle = \langle\nabla \psi^b_0\cdot Q_{\varphi_j}\cdot\nabla\psi^b_{1,-}\rangle + \langle\nabla \psi^b_{2,-}\cdot Q_{\varphi_j}\cdot\nabla\psi^b_{1,+}\rangle + \langle\nabla \psi^b_{2,-}\cdot Q_{\varphi_j}\cdot\nabla\psi^b_{2,+}\rangle.
		\end{equation}
        Applying \cref{eq:opnorm} to this expression, followed by Young's inequality, we bound the first right-hand term in~\cref{eq:thirdterm1} as follows:
		\begin{equation}
			\begin{aligned}
				\left|2\gamma_2\Re [A_{ij}\ovl{c_i}\langle\varphi_j,\psi^b,\nabla^2\psi^b\rangle]\right| 
        &\leq 2|\gamma_2|(|c_1|\|Q_{\varphi_2}\|_\mathrm{op} + |c_2|\|Q_{\varphi_1}\|_\mathrm{op})(\sqrt{2}q_0q_1 + q_1q_2 + q_2^2) \\
		&\leq |\gamma_2|(\sqrt{2}q_0q_1 + q_1q_2 + q_2^2)\\
        & \qquad\times \left[\left(\eps_1+\tfrac{1}{\eps_1}|c_1|^2\right)\|Q_{\varphi_2}\|_\mathrm{op} + \left(\eps_2+\tfrac{1}{\eps_2}|c_2|^2\right)\|Q_{\varphi_1}\|_\mathrm{op}\right]
		\end{aligned}
		\end{equation}
		for any $\eps_1,\eps_2>0$, where we have used the identity $\|\nabla\psi^b_0\|=\sqrt{2}q_0$. One could sharpen this bound by projecting $Q_{\varphi_j}$ to particular subspaces $\cH_k\subset\cH_\CC$ before taking its operator norm, but we do not pursue this improvement here.

		We similarly estimate the second term on the right-hand side of~\cref{eq:thirdterm1} as
		\begin{equation}
			\begin{aligned}
				\left|2\gamma_2\Re [A_{ij}\ovl{c_i}\langle \Psi,\nabla\varphi_j,\nabla\psi^b\rangle]\right|&\leq2|\gamma_2|\left|\langle A_{ij}\ovl{c_i}\{\Psi,\nabla\varphi_j\},\nabla\psi^b_{1,-}\rangle\right|\\
				&\leq 2|\gamma_2|q_1\|A_{ij}\ovl{c_i}\{\Psi,\nabla\varphi_j\}\|\\
				&\leq \eps_0|\gamma_2|q_1^2 + \frac{|\gamma_2|}{\eps_0} A_{ij}A_{k\ell}\langle \{\Psi,\nabla\varphi_j\},\{\Psi,\nabla\ovl{\varphi_k}\}\rangle \ovl{c_i}c_\ell,
			\end{aligned}
		\end{equation}
		and the third term as
		\begin{equation}
			\left|\gamma_2\Re [A_{ij}\ovl{c_i}\langle Z,\varphi_j,\psi^b\rangle]\right|\leq \eps_7|\gamma_2|q_1^2 + \frac{|\gamma_2|}{4\eps_7} A_{ij}A_{k\ell}\langle \{Z,\varphi_j\},\{Z,\ovl{\varphi_k}\}\rangle \ovl{c_i}c_\ell.
		\end{equation}
		Notably, terms with $\eps_7$ can be deleted in the case of Couette flow (but not Poiseuille flow), since $Z$ is constant. Combining the three preceding estimates gives our bound on the third term of $[\dot{V}]_b$:
        \begin{equation}
        \label{eq:bound3}
            \begin{aligned}
                \big|2\gamma_2\Re &(A_{ij}\ovl{c_i}d_j[\psi^b])\big| \\
                &\leq |\gamma_2|(\sqrt{2}q_0q_1 + q_1q_2 + q_2^2)\left[\left(\eps_1+\tfrac{1}{\eps_1}|c_1|^2\right)\|Q_{\varphi_2}\|_\mathrm{op} + \left(\eps_2+\tfrac{1}{\eps_2}|c_2|^2\right)\|Q_{\varphi_1}\|_\mathrm{op}\right]\\
                &\qquad + \eps_0|\gamma_2|q_1^2 + \frac{|\gamma_2|}{\eps_0} A_{ij}A_{k\ell}\langle \{\Psi,\nabla\varphi_j\},\{\Psi,\nabla\ovl{\varphi_k}\}\rangle \ovl{c_i}c_\ell\\
                &\qquad + \eps_7|\gamma_2|q_1^2 + \frac{|\gamma_2|}{4\eps_7} A_{ij}A_{k\ell}\langle \{Z,\varphi_j\},\{Z,\ovl{\varphi_k}\}\rangle \ovl{c_i}c_\ell.
            \end{aligned}
        \end{equation}

		\subsection[Fifth term of Vb]{Fifth term of $[\dot{V}]_b$}
        
		For the fifth term of $[\dot{V}]_b$, we first rewrite the integral by expanding $\psi$ into components and using orthogonality:
		\begin{equation}
        \langle\nabla\psi^a\cdot Q_{\varphi_j}\cdot\nabla\psi^b\rangle = \ovl{c_k}\langle\nabla\varphi_k\cdot Q_{\varphi_j}\cdot\nabla\psi^b_{2,-}\rangle + c_k\langle\nabla\ovl{\varphi_k}\cdot Q_{\varphi_j}\cdot\nabla\psi^b_{0}\rangle + 2c_0\langle\nabla\varphi_0\cdot Q_{\varphi_j}\cdot\nabla\psi^b_{1,-}\rangle.
		\end{equation}
		Then we estimate the full fifth term of $[\dot{V}]_b$ by
		\begin{equation}
			\begin{aligned}
				\big|2&\gamma_2\Re A_{ij}\ovl{c_i}\langle\nabla\psi^a\cdot Q_{\varphi_j}\cdot\nabla\psi^b\rangle\big| \leq 2|\gamma_2||A_{ij}\ovl{c_i}\langle\nabla\psi^a\cdot Q_{\varphi_j}\cdot\nabla\psi^b\rangle|\\
				&\leq 2|\gamma_2| |c_1||c_k|(\|Q_{\varphi_2}\cdot\nabla\varphi_k\| q_2 + \|Q_{\varphi_2}\cdot\nabla\ovl{\varphi_k}\| \sqrt{2}q_0) + 4\gamma_2|c_1||c_0|\|Q_{\varphi_2}\cdot\nabla\varphi_0\|q_1 \\
				&\qquad + 2|\gamma_2| |c_2||c_k|(\|Q_{\varphi_1}\cdot\nabla\varphi_k\| q_2 + \|Q_{\varphi_1}\cdot\nabla\ovl{\varphi_k}\| \sqrt{2}q_0) + 4|\gamma_2||c_2||c_0|\|Q_{\varphi_1}\cdot\nabla\varphi_0\|q_1\\
				&= 2|\gamma_2 ||c_1|^2(\|Q_{\varphi_2}\cdot\nabla\varphi_1\| q_2 + \|Q_{\varphi_2}\cdot\nabla\ovl{\varphi_1}\|\sqrt{2}q_0)\\
				&\qquad + 2|\gamma_2| |c_2|^2(\|Q_{\varphi_1}\cdot\nabla\varphi_2\| q_2 + \|Q_{\varphi_1}\cdot\nabla\ovl{\varphi_2}\|\sqrt{2}q_0)\\
				&\qquad + 2|\gamma_2| |c_1||c_2|\big[(\|Q_{\varphi_1}\cdot\nabla\varphi_1\| +\|Q_{\varphi_2}\cdot\nabla\varphi_2\|)q_2 \\
				&\qquad\qquad\qquad\qquad\qquad\qquad+ (\|Q_{\varphi_1}\cdot\nabla\ovl{\varphi_1}\|+\|Q_{\varphi_2}\cdot\nabla\ovl{\varphi_2}\|)\sqrt{2}q_0)\big]\\
				&\qquad + 4|\gamma_2||c_1||c_0|\|Q_{\varphi_2}\cdot\nabla\varphi_0\|q_1 + 4|\gamma_2||c_2||c_0|\|Q_{\varphi_1}\cdot\nabla\varphi_0\|q_1.
			\end{aligned}
		\end{equation}
        We want $|c_j|$ to only appear squared in the bound, so that the full expression is polynomial in the complex variable $c_j$. Young's inequality gives
		\begin{equation}
			\begin{gathered}
				|c_1||c_0|\leq \frac{1}{2\eps_5}|c_1|^2 + \frac{\eps_5}{2}|c_0|^2,\qquad |c_2||c_0|\leq \frac{1}{2\eps_6}|c_2|^2 + \frac{\eps_6}{2}|c_0|^2,\\ |c_1||c_2|\leq \frac{1}{2\eps_4}|c_1|^2 + \frac{\eps_4}{2}|c_2|^2,
			\end{gathered}
		\end{equation}
		with which we arrive at our final bound:
        \begin{equation}
        \label{eq:bound5}
        \begin{aligned}
			\big|2\gamma_2\Re A_{ij}\ovl{c_i}&\langle\nabla\psi^a\cdot Q_{\varphi_j}\cdot\nabla\psi^b\rangle\big| \\
            &\leq 2|\gamma_2| |c_1|^2(\|Q_{\varphi_2}\cdot\nabla\varphi_1\| q_2 + \|Q_{\varphi_2}\cdot\nabla\ovl{\varphi_1}\|\sqrt{2}q_0)\\
			&\qquad + 2|\gamma_2| |c_2|^2(\|Q_{\varphi_1}\cdot\nabla\varphi_2\| q_2 + \|Q_{\varphi_1}\cdot\nabla\ovl{\varphi_2}\|\sqrt{2}q_0)\\
			&\qquad + 2|\gamma_2| \left(\tfrac{1}{2\eps_4}|c_1|^2 + \frac{\eps_4}{2}|c_2|^2\right)\big[(\|Q_{\varphi_1}\cdot\nabla\varphi_1\| +\|Q_{\varphi_2}\cdot\nabla\varphi_2\|)q_2 \\
			&\qquad \hspace{2in}+ (\|Q_{\varphi_1}\cdot\nabla\ovl{\varphi_1}\|+\|Q_{\varphi_2}\cdot\nabla\ovl{\varphi_2}\|)\sqrt{2}q_0\big]\\
			&\qquad + 4|\gamma_2|\left(\tfrac{1}{2\eps_5}|c_1|^2 + \frac{\eps_5}{2}|c_0|^2\right)\|Q_{\varphi_2}\cdot\nabla\varphi_0\|q_1\\
			&\qquad + 4|\gamma_2|\left(\tfrac{1}{2\eps_6}|c_2|^2 + \frac{\eps_6}{2}|c_0|^2\right)\|Q_{\varphi_1}\cdot\nabla\varphi_0\|q_1.
            \end{aligned}
		\end{equation}

        \subsection{Coefficients of the full polynomial}
        \label{sec:S coeff}
        
        Combining the bounds \cref{eq:bound1,eq:bound2,eq:bound3,eq:bound4,eq:bound5} on the five terms of $[\dot V]_b$, we finally obtain a polynomial bound of the form $\dot{V}[\psi]\le S(c_j,q_n)$. This $S$ takes the form given by~\cref{eq:fullpoly}, where the coefficients $B_{ij}$, $C_{ij}$, $D_{ij}$, $\kappa_n$, $\eta_j$, $\beta_{n,m}$, and $\omega$ are defined as follows. In these definitions, the indices $k$ and $\ell$ are summed over whenever they appear, but $i$ and $j$ are not.
        \begin{subequations}
        \begin{equation}
        \begin{aligned}
            B_{ij} &= A_{ij}(\lambda_j + \lambda_i) + A_{ik}[\tfrac{1}{2}\langle\varphi_k,\ovl{\varphi_j},Z\rangle + \langle\Psi,\nabla\varphi_k,\nabla\ovl{\varphi_j}\rangle ] - A_{kj}[\tfrac{1}{2}\langle\varphi_i,\ovl{\varphi_k},Z\rangle + \langle\Psi,\nabla\varphi_i,\nabla\ovl{\varphi_k}\rangle ] \\
            &\qquad + \frac{1}{\eps_0} A_{i\ell}A_{kj}\langle \{\Psi,\nabla\varphi_\ell\},\{\Psi,\nabla\ovl{\varphi_k}\}\rangle + \frac{1}{4\eps_7} A_{i\ell}A_{kj}\langle \{Z,\varphi_\ell\},\{Z,\ovl{\varphi_k}\}\rangle,
        \end{aligned}
        \end{equation}
        \begin{equation}
        C_{ij} = 2\langle\nabla\ovl{\varphi_j}\cdot Q_{\varphi_0}\cdot\nabla\varphi_i\rangle + 2 \frac{1}{\eps_3}\langle \ovl{Q_{\varphi_0}}\cdot\nabla\varphi_i,Q_{\varphi_0}\cdot\nabla\ovl{\varphi_j}\rangle,
        \end{equation}
        \begin{equation}
        D_{ij} = 2A_{ik}\langle\nabla\ovl{\varphi_j}\cdot Q_{\varphi_k}\cdot\nabla\varphi_0\rangle + 2A_{kj}\langle\nabla\varphi_i\cdot Q_{\ovl{\varphi_k}}\cdot\nabla\varphi_0\rangle,
        \end{equation}
        \begin{equation}
        \beta_{00} = \beta_{20} = 0,\qquad \beta_{10} = 4\frac{\eps_6}{2}\|Q_{\varphi_1}\cdot\nabla\varphi_0\| + 4\frac{\eps_5}{2}\|Q_{\varphi_2}\cdot\nabla\varphi_0\|,
        \end{equation}
        \begin{equation}
        \beta_{11} = 4\frac{1}{2\eps_5}\|Q_{\varphi_2}\cdot\nabla\varphi_0\|,\qquad \beta_{12} = 4\frac{1}{2\eps_6}\|Q_{\varphi_1}\cdot\nabla\varphi_0\|,
        \end{equation}
        \begin{equation}
        \beta_{01} = 2\sqrt{2}\, \|Q_{\varphi_2}\cdot\nabla\ovl{\varphi_1}\| + 2\sqrt{2}\frac{1}{2\eps_4}(\|Q_{\varphi_1}\cdot\nabla\ovl{\varphi_1}\|+\|Q_{\varphi_2}\cdot\nabla\ovl{\varphi_2}\|),
        \end{equation}
        \begin{equation}
        \beta_{02} = 2\sqrt{2}\, \|Q_{\varphi_1}\cdot\nabla\ovl{\varphi_2}\| + 2\sqrt{2}\frac{\eps_4}{2}(\|Q_{\varphi_1}\cdot\nabla\ovl{\varphi_1}\|+\|Q_{\varphi_2}\cdot\nabla\ovl{\varphi_2}\|),
        \end{equation}
        \begin{equation}
        \beta_{21} = 2\, \|Q_{\varphi_2}\cdot\nabla\varphi_1\| + 2\frac{1}{2\eps_4}(\|Q_{\varphi_1}\cdot\nabla\varphi_1\| +\|Q_{\varphi_2}\cdot\nabla\varphi_2\|),
        \end{equation}
        \begin{equation}
        \beta_{22} = 2\, \|Q_{\varphi_1}\cdot\nabla\varphi_2\| + 2\frac{\eps_4}{2}(\|Q_{\varphi_1}\cdot\nabla\varphi_1\| +\|Q_{\varphi_2}\cdot\nabla\varphi_2\|),
        \end{equation}

        \begin{equation}
        \tau_0=0,\qquad\tau_1=4\pi L^{-1/2} + 2\eps_3,\qquad \tau_2=4\pi L^{-1/2},
        \end{equation}
        \begin{equation}
        \eta_0 = (\eps_1\|Q_{\varphi_2}\|_\mathrm{op} + \eps_2\|Q_{\varphi_1}\|_\mathrm{op}),\qquad \eta_1 = \frac{1}{\eps_1}\|Q_{\varphi_2}\|_\mathrm{op},\qquad \eta_2 = \frac{1}{\eps_2}\|Q_{\varphi_1}\|_\mathrm{op}.
        \end{equation}
        \begin{equation}
        \omega = \eps_0+\eps_7
        \end{equation}
        \end{subequations}

        \section{Proof of \cref{prop:young}}
        \label{app:youngs}

        We here prove \cref{prop:young}, which demonstrates that Young's inequality always gives an optimal polynomial bound on the non-polynomial terms in our analysis. 
        \begin{proof}
            The first inequality in~\cref{eq:youngs prop} is Young's inequality, which holds for all $\eps>0$. We must prove that there exists $\eps>0$ for which the second inequality holds. Write $p(\vec{z}) = \vec{z}^t\vec{P}\vec{z}$, $\vec{A}^t\vec{A} = \vec{Q}$, and $\vec{B}^t\vec{B}=\vec{R}$, and let $r=\op{rank}\vec{P}$. Without loss of generality, we can choose coordinates such that
            \[\vec{P}=\begin{pmatrix}
                \vec{I}_{r} & 0\\ 0 & 0
            \end{pmatrix},\]
            scaling $\vec{P}$, $\vec{A}$, and $\vec{B}$ as appropriate. We can also restrict attention to the first $r$ coordinates, so that \cref{eq:generic_poly} becomes
            \begin{equation}\label{eq:generic_poly_redone}
            |\vec{z}|^2 \geq |\vec{A}\vec{z}|\,|\vec{B}\vec{z}|
            \end{equation}
            Now, write
            \begin{equation}
                \vec{M}_\eps = \frac{\eps}{2}\vec{Q} + \frac{1}{2\eps}\vec{R},\qquad \sigma(\eps)=\max_{|\vec{z}|=1}\vec{z}^t\vec{M}_\eps\vec{z}.
            \end{equation}
            Since $\sigma(\eps)$ diverges when $\eps\to 0$ and when $\eps\to\infty$, it takes a minimum at some finite $\eps_*>0$. Let $\vec{z}_\eps\in\RR^r$ be the leading (unit) eigenvector of $\vec{M}_{\eps}$, and $\vec{z}_*=\vec{z}_{\eps_*}$; even in the (non-generic) case that the eigenspace is degenerate at $\eps_*$, one can choose $\vec{z}_*$ as the endpoint of a smooth path $\vec{z}_\eps$ for $\eps_*+\delta>\eps>\eps_*$ and sufficiently small $\delta>0$. We differentiate
            \begin{equation}
            \begin{aligned}
                0&=\frac{d}{d\eps}\bigg|_{\eps=\eps_*}\sigma(\eps) \\
                &= \frac{d}{d\eps}\bigg|_{\eps=\eps_*}\vec{z}_\eps^t\vec{M}_{\eps}\vec{z}_\eps \\
                &= 2\sigma(\eps_*)\vec{z}_*^t\frac{d}{d\eps}\bigg|_{\eps=\eps_*}\vec{z}_\eps + \frac{1}{2}\vec{z}^t_*\vec{Q}\vec{z}_* - \frac{1}{2\eps^2_*}\vec{z}^t_* \vec{R}\vec{z}_*\\
                &= \frac{1}{2}\vec{z}^t_*\vec{Q}\vec{z}_* - \frac{1}{2\eps^2_*}\vec{z}^t_* \vec{R}\vec{z}_*,
                \end{aligned}
            \end{equation}
            noting that $|\vec{z}_\eps|=1$ for all $\eps$. In particular, neither term can vanish, or they both must, and one would have $\sigma(\eps_*)=0$. We can thus write
            \[\eps_* = \sqrt{\frac{\vec{z}^t_*\vec{R}\vec{z}_*}{\vec{z}^t_*\vec{Q}\vec{z}_*}}\qquad \implies\qquad \sigma(\eps_*)=\vec{z}_*^t\vec{M}_{\eps_*}\vec{z}_* = |\vec{A}\vec{z}_*|\,|\vec{B}\vec{z}_*|\leq 1,\]
            employing \cref{eq:generic_poly_redone} in the final step. The second inequality in~\cref{eq:youngs prop} then follows, which proves the main claim. The alternate version of the proposition, where $\vec{B}\vec{z}$ is replaced with $1$, is proved analogously.
        \end{proof}
		
		\bibliographystyle{jfm}
		\bibliography{GlobalStability}
		
		%% End of file `jfm.bib'.

	\end{document}

%% file: dave_notation.tex
\usepackage{mathtools, amsfonts, amssymb, amstext,graphicx,xcolor,tikz,tkz-euclide, booktabs, tensor, empheq,subcaption,lipsum,float,empheq,tcolorbox,enumitem}
\usepackage{algpseudocode}

\usepackage[capitalize]{cleveref}
\crefformat{equation}{(#2#1#3)}

\newcommand{\eps}{\varepsilon}
\newcommand{\ovl}[1]{{\overline{#1}}}
\newcommand{\op}{\operatorname}
\renewcommand{\d}{\mathrm{d}}

\newcommand{\CC}{\mathbb{C}}

\newcommand{\RR}{\mathbb{R}}

\newcommand{\ZZ}{\mathbb{Z}}

\newcommand{\cE}{\mathcal{E}}

\newcommand{\cH}{\mathcal{H}}

\newcommand{\cL}{\mathcal{L}}

\newcommand{\cR}{\mathcal{R}}

\newcommand{\nRe}{\mathit{Re}} % Reynolds
\renewcommand{\vec}[1]{\mathbf{#1}}

\newtheorem{lemma}{Lemma}
\newtheorem{proposition}{Proposition}

\newtheorem{definition1}{Definition}[section]
\newtheorem{remark1}{Remark}[section]

\makeatletter
\newenvironment{definition}[1][]{
  \begin{definition1}\normalfont{}#1
}{\end{definition1}}
\newenvironment{remark}[1][]{
  \begin{remark1}\normalfont{}#1
}{\end{remark1}}
\makeatother

\newtcolorbox{remarkbox}{
  colback=white,
  colframe=black,
  boxrule=1pt,
  arc=0pt,
  left=6pt,
  right=6pt,
  top=0pt,
  bottom=6pt,
}

\let\oldremark\remark
\let\endoldremark\endremark

\renewenvironment{remark}
  {\begin{remarkbox}\oldremark}
  {\endoldremark\end{remarkbox}}